\documentclass[twocolumn, journal]{IEEEtran}
\IEEEoverridecommandlockouts

\usepackage{color}
\usepackage{bm}
\usepackage{graphicx}
\usepackage{amsmath}
\usepackage{amssymb}
\usepackage{algorithm}
\usepackage{algorithmic}
\usepackage{multirow}
\usepackage{booktabs}
\usepackage{array}
\usepackage{amsthm}
\usepackage{lipsum}
\usepackage{enumerate}
\usepackage{stfloats}
\usepackage{subfigure}
\usepackage{cases}
\usepackage{diagbox}
\usepackage{threeparttable}
\usepackage{cite}

\newtheorem{remark}{Remark}

\newtheorem{lemma}{Lemma}
\newtheorem{corollary}{Corollary}
\newtheorem{assump}{Assumption}

\title{Beyond Diagonal Reconfigurable Intelligent Surfaces: A Multi-Sector Mode Enabling Highly Directional Full-Space Wireless Coverage}

\author{Hongyu Li,~\IEEEmembership{Student Member,~IEEE}, Shanpu Shen,~\IEEEmembership{Member,~IEEE}, and Bruno Clerckx,~\IEEEmembership{Fellow,~IEEE}
\thanks{Manuscript received; \textit{(Corresponding author: Shanpu Shen).}}
\thanks{H. Li is with the Department of Electrical and Electronic Engineering, Imperial College London, London SW7 2AZ, U.K. (email: c.li21@imperial.ac.uk).}\\
\thanks{S. Shen is with the Department of Electronic and Computer Engineering, The Hong Kong University of Science and Technology, Clear Water Bay, Kowloon, Hong Kong (email: sshenaa@connect.ust.hk).}
\thanks{B. Clerckx is with the Department of Electrical and Electronic Engineering, Imperial College London, London SW7 2AZ, U.K. and with Silicon Austria Labs (SAL), Graz A-8010, Austria (email: b.clerckx@imperial.ac.uk; bruno.clerckx@silicon-austria.com).}
}

\begin{document}

\maketitle
\thispagestyle{empty}
\begin{abstract}

Reconfigurable intelligent surface (RIS) has gained much traction due to its potential to manipulate the propagation environment via nearly-passive reconfigurable elements. 
In our previous work, we have analyzed and proposed a beyond diagonal RIS (BD-RIS) model, which is not limited to traditional diagonal phase shift matrices, to unify different RIS modes/architectures.  
In this paper, we create a new branch of BD-RIS supporting a multi-sector mode.
A multi-sector BD-RIS is modeled as multiple antennas connected to a multi-port group-connected reconfigurable impedance network. More specifically, antennas are divided into $L$ ($L \ge 2$) sectors and arranged as a polygon prism with each sector covering $1/L$ space.  
Different from the recently introduced concept of intelligent omni-surface (or simultaneously transmitting and reflecting RIS), the multi-sector BD-RIS not only achieves a full-space coverage, but also has significant performance gains thanks to the highly directional beam of each sector. 
We derive the constraint of the multi-sector BD-RIS and the corresponding channel model taking into account the relationship between antenna beamwidth and gain. 
With the proposed model, we first derive the scaling law of the received signal power for a multi-sector BD-RIS-assisted single-user system. 
We then propose efficient beamforming design algorithms to maximize the sum-rate of the multi-sector BD-RIS-assisted multiuser system.  
Simulation results verify the effectiveness of the proposed design and demonstrate the performance enhancement of the proposed multi-sector BD-RIS.
\end{abstract}


\section{Introduction}

The forthcoming next-generation 6G is expected to address 5G drawbacks, improve existing techniques, and lead new technological trends \cite{tataria20216g}.  
Reconfigurable intelligent surface (RIS) is regarded as one of the promising techniques which will drive the 6G architectural evolution \cite{saad2019vision}. 
RIS, usually consisting of multiple nearly-passive reconfigurable elements with flexibly adjustable amplitudes and/or phases, enables adaptive manipulations of the scattered signals \cite{di2020smart}. 
The ``nearly-passive'' and ``reconfigurable'' characteristics of the RIS show its potential to shape the propagation environment as intended in an easily deployed way with a reduced cost, power consumption, and weight \cite{di2019smart}.

\begin{figure}
    \centering
    \includegraphics[width=0.48\textwidth]{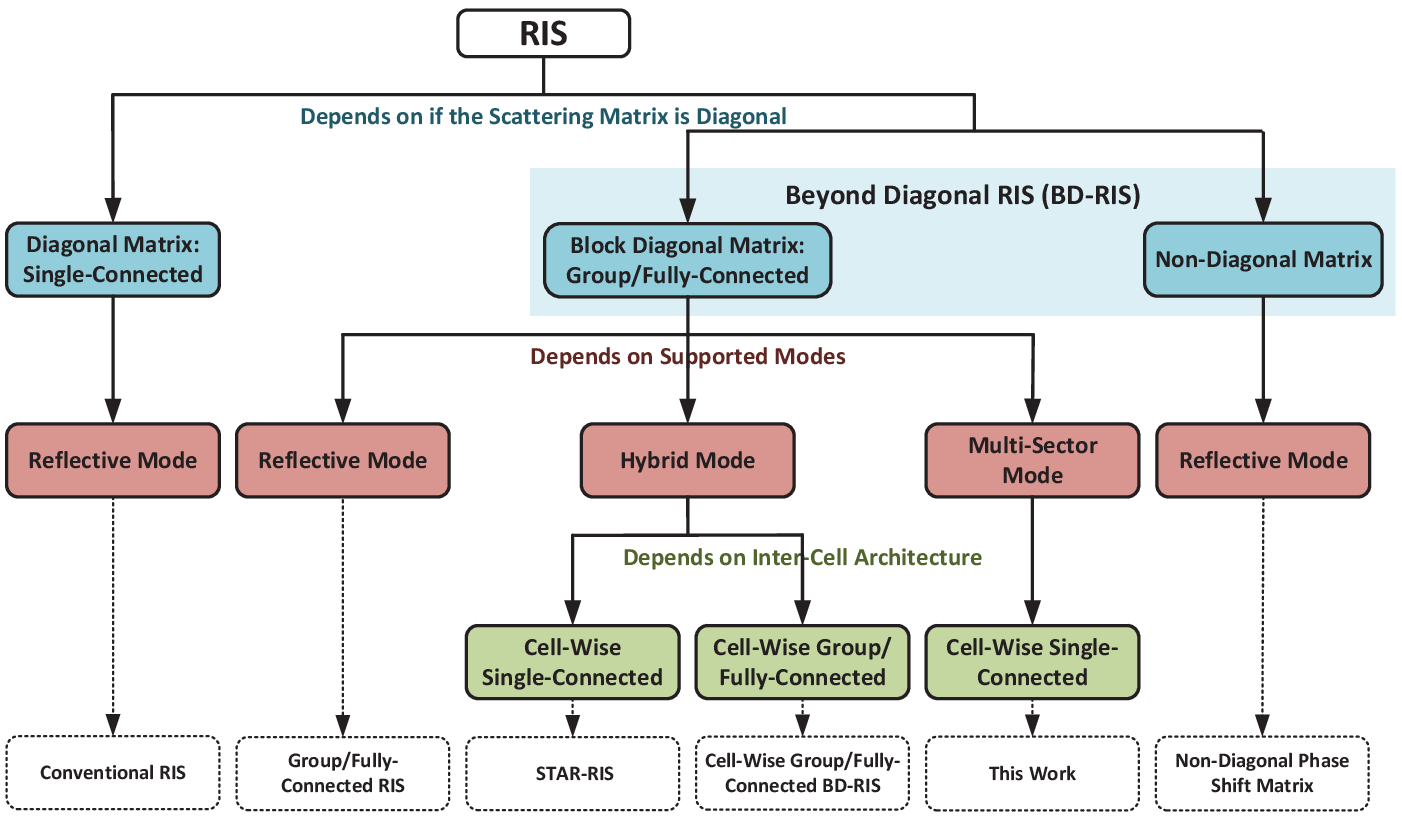}
    \caption{RIS classification tree.}
    \label{fig:RIS_tree}
\end{figure}

Given the appealing potential of RIS, extensive work \cite{guo2020weighted,di2020hybrid,huang2019reconfigurable,wu2019intelligent,wu2019beamforming,zheng2021double,bansal2021rate,feng2022waveform,zhao2021irs,liu2022joint,xu2020resource,gopi2020intelligent,abeywickrama2020intelligent,li2021intelligent,ozdogan2019intelligent,tang2020wireless,shen2021,nerini2021reconfigurable,li2022reconfigurable,xu2021star,zhang2022intelligent,li2022} has been devoted to modeling and optimizations for RIS-aided systems. To clearly classify the existing works, we categorize RIS, according to the mathematical characteristic of its scattering matrix, into diagonal RIS where the scattering matrix is diagonal, and beyond diagonal RIS (BD-RIS) where the scattering matrix is not restricted to be diagonal, as shown in Fig. \ref{fig:RIS_tree}. 
In the following, we will classify the existing studies \cite{guo2020weighted,di2020hybrid,huang2019reconfigurable,wu2019intelligent,wu2019beamforming,zheng2021double,bansal2021rate,feng2022waveform,zhao2021irs,liu2022joint,xu2020resource,gopi2020intelligent,abeywickrama2020intelligent,li2021intelligent,ozdogan2019intelligent,tang2020wireless,shen2021,nerini2021reconfigurable,li2022reconfigurable,xu2021star,zhang2022intelligent,li2022} by mapping them to different branches in Fig. \ref{fig:RIS_tree}.

For conventional RIS \cite{guo2020weighted,di2020hybrid,huang2019reconfigurable,wu2019intelligent,wu2019beamforming,zheng2021double,bansal2021rate,feng2022waveform,zhao2021irs,liu2022joint,xu2020resource,gopi2020intelligent,abeywickrama2020intelligent,li2021intelligent,ozdogan2019intelligent,tang2020wireless}, each element is connected to its own reconfigurable load, but is disconnected from other elements. Therefore, the scattering matrix of conventional RIS \cite{guo2020weighted,di2020hybrid,huang2019reconfigurable,wu2019intelligent,wu2019beamforming,zheng2021double,bansal2021rate,feng2022waveform,zhao2021irs,liu2022joint,xu2020resource,gopi2020intelligent,abeywickrama2020intelligent,li2021intelligent,ozdogan2019intelligent,tang2020wireless} is modeled as a simple \textit{diagonal phase shift matrix}. 
This kind of RIS can only reflect signals from one side of the RIS towards the receiver, and thus only supports the reflective mode. 
With diagonal RIS, beamforming designs focusing on different metrics, such as rate maximization \cite{guo2020weighted,di2020hybrid}, energy-efficiency maximization \cite{huang2019reconfigurable}, power minimization \cite{wu2019intelligent,wu2019beamforming}, and max-min fairness \cite{zheng2021double}, have been developed for various RIS-aided wireless communication systems with continuous phase shifts \cite{guo2020weighted}, \cite{huang2019reconfigurable,wu2019intelligent,wu2019beamforming,zheng2021double} and discrete phase shifts \cite{di2020hybrid,wu2019beamforming}.  
In addition, the integration of RIS and other techniques has also been widely investigated, such as rate-splitting multiple access (RSMA) \cite{bansal2021rate}, wireless power transfer (WPT) \cite{feng2022waveform}, simultaneous wireless information and power transfer (SWIPT) \cite{zhao2021irs}, integrated sensing and communication (ISAC) \cite{liu2022joint}, cognitive radio \cite{xu2020resource}, and index modulation \cite{gopi2020intelligent}. 
The aforementioned works \cite{guo2020weighted,di2020hybrid,huang2019reconfigurable,wu2019intelligent,wu2019beamforming,zheng2021double,bansal2021rate,feng2022waveform,zhao2021irs,liu2022joint,xu2020resource,gopi2020intelligent} assume an ideal RIS phase shift model where its diagonal elements have unit modulus, while only a few results focus on the modeling of RIS accounting for the lumped inductance and capacitance \cite{abeywickrama2020intelligent,li2021intelligent}.
In addition, the near- and far-field free-space path loss models for RIS-aided systems are also theoretically derived \cite{ozdogan2019intelligent,tang2020wireless} and experimentally verified \cite{tang2020wireless}.

Different from conventional RISs with diagonal scattering matrices, BD-RISs have scattering matrices not limited to diagonal ones, which can provide more flexible beam management.   
In existing works, there are two branches belonging to BD-RIS supporting the reflective mode, as shown in Fig. \ref{fig:RIS_tree}.   
On one hand, \cite{shen2021} for the first time models RIS based on the scattering parameter network. This work goes beyond diagonal phase shift models \cite{guo2020weighted,di2020hybrid,huang2019reconfigurable,wu2019intelligent,wu2019beamforming,zheng2021double,bansal2021rate,feng2022waveform,zhao2021irs,liu2022joint,xu2020resource,gopi2020intelligent,abeywickrama2020intelligent,li2021intelligent,ozdogan2019intelligent,tang2020wireless} (with a single-connected architecture), and proposes more general group/fully-connected architectures with block-diagonal matrices.  
Different from the single-connected architecture that only adjusts the phase shifts of incident waves, the group/fully-connected architectures enable both amplitude and phase shift manipulations, therefore providing additional beam control flexibility.  
With the proposed architectures in \cite{shen2021}, codebooks for group/fully-connected RIS with discrete-value impedance networks are designed in \cite{nerini2021reconfigurable}, showing that the codebook size for fully-connected RIS can be much reduced compared to that in conventional RIS. 
On the other hand, another novel RIS architecture \cite{li2022reconfigurable} is proposed under the BD-RIS branch, where the signal impinging on one reconfigurable element can be reflected from another one, yielding a non-diagonal phase shift matrix that enables more flexible beam manipulation.

The limitation of the abovementioned RIS models in \cite{guo2020weighted,di2020hybrid,huang2019reconfigurable,wu2019intelligent,wu2019beamforming,zheng2021double,bansal2021rate,feng2022waveform,zhao2021irs,liu2022joint,xu2020resource,gopi2020intelligent,abeywickrama2020intelligent,li2021intelligent,ozdogan2019intelligent,tang2020wireless,shen2021,nerini2021reconfigurable,li2022reconfigurable} is that incident signals can only be reflected to one side of RIS due to the uniform antenna array arrangement, such that half-space is out of the coverage of RIS.
To address this issue, the concept of simultaneously transmitting and reflecting RIS (STAR-RIS) \cite{xu2021star} or intelligent omni-surface (IOS) \cite{zhang2022intelligent} has been introduced recently. 
STAR-RIS is a dual-functional device enabling signal manipulations of both reflection and transmission, i.e., supporting the hybrid mode as illustrated in Fig. \ref{fig:RIS_tree}, to achieve full-space coverage. 
However, STAR-RIS \cite{xu2021star,zhang2022intelligent} is a special case of back to back placed antenna arrays connected to group-connected reconfigurable impedance network, while different architectures of BD-RIS supporting the hybrid mode have not been investigated.
This motivates our recent work \cite{li2022}, where a general BD-RIS model unifying reflective/transmissive/hybrid modes and single/group/fully-connected architectures is proposed and analyzed. 
Specifically, we first show that STAR-RIS is a specific application of group-connected reconfigurable impedance network with group size equal to 2. Then we go beyond STAR-RIS and propose BD-RIS with hybrid mode and more general group/fully-connected architectures, which achieve better performance than STAR-RIS, as shown in Fig. \ref{fig:RIS_tree}. 

While the proposed BD-RIS \cite{li2022} with hybrid mode and group/fully-connected architectures achieves significant performance improvement, it is achieved at the expense of increasing circuit complexity. This motivates us to explore different architectures of BD-RIS together with different antenna arrangements to achieve fully-space coverage, enhance performance, and simplify the circuit design. Therefore, in this work, we propose a novel multi-sector BD-RIS as shown in Fig. \ref{fig:RIS_tree}. More specifically, we have the following contributions.

\textit{First}, we, for the first time, propose a multi-sector BD-RIS, which is modeled as a polygon as shown in Fig. \ref{fig:assumption}(a). 
The proposed model is general enough to include STAR-RIS/IOS as a special instance and achieve full-space coverage, but provides significant performance enhancements compared to IOS/STAR-RIS and conventional RIS, benefiting from the group-connected reconfigurable impedance network and antenna array arrangements.
With full-space coverage, highly directional beams, and adaptive structures (based on different numbers of sectors), multi-sector BD-RIS can be flexibly deployed into various situations and is particularly favorable to millimeter wave/Terahertz scenarios and cell-free networks.

\textit{Second}, we derive constraints of multi-sector BD-RIS when its reconfigurable elements have either continuous amplitudes/phase shifts or discrete values selected from finite codebooks. We also derive the corresponding channel model by taking into account the impact of beamwidth and antenna gain.

\textit{Third}, with the proposed model, we derive the scaling law of the received signal power as a function of the number of sectors for a multi-sector BD-RIS-assisted single user single-input single-output (SU-SISO) system to show the benefits of using multi-sector BD-RIS.

\textit{Fourth}, we apply the proposed multi-sector BD-RIS in multiuser multiple-input single-output (MU-MISO) systems. Specifically, we propose efficient algorithms to maximize the sum-rate performance of the multi-sector BD-RIS-assisted MU-MISO system. 
We first derive the closed-form solution of the multi-sector BD-RIS when its scattering matrix has continuous values. Then we propose an efficient algorithm to obtain the discrete solution.

\textit{Fifth}, we evaluate the sum-rate performance of the multi-sector BD-RIS-aided MU-MISO system. 
Numerical results show that with the same sum-rate performance requirement, the number of antennas for the whole multi-sector BD-RIS can be reduced by increasing the number of sectors. 
In addition, the resolution of the finite codebooks for the multi-sector BD-RIS with discrete values can be reduced with increasing number of sectors, which is beneficial for controlling cost for practical RIS realizations.

\textit{Organization:} Section II proposes a novel multi-sector BD-RIS model and derives the corresponding channel model. 
Section III lists potential applications of the multi-sector BD-RIS. 
Section IV derives the scaling law for the multi-sector BD-RIS-assisted SU-SISO system. Section V provides the beamforming design for the multi-sector BD-RIS-assisted MU-MISO system. 
Section VI evaluates the performance of the proposed design and Section VII concludes this work.

\textit{Notations}:
Boldface lower- and upper-case letters indicate column vectors and matrices, respectively.
$\mathbb{C}$, $\mathbb{R}$, and $\mathbb{Z}$ denote the set of complex numbers, real numbers, and integers, respectively.
$\mathbb{E}\{\cdot\}$ is statistical expectation.
$(\cdot)^\ast$, $(\cdot)^T$, $(\cdot)^H$, and $(\cdot)^{-1}$ denote the conjugate, transpose, conjugate-transpose operations, and inversion, respectively.
$\mathcal{CN}(0,\sigma^2)$ denotes the circularly symmetric complex Gaussian distribution with zero mean and covariance $\sigma^2$.
$\mathcal{U}(a,b)$ denotes the continuous uniform distribution with boundaries $a$ and $b$. 
$\Re \{ \cdot \}$ denotes the real part of a complex number.
$\mathbf{I}_L$ indicates an $L \times L$ identity matrix.
$\mathbf{0}$ denotes an all-zero matrix.
$\| \mathbf{A} \|_F$ denotes the Frobenius norm of matrix $\mathbf{A}$.
$|a|$ denotes the norm of variable $a$.
$\mathsf{diag}(\cdot)$ denotes a diagonal matrix.
$\jmath = \sqrt{-1}$ denotes imaginary unit.
$\mathsf{Tr}\{\cdot\}$ denotes the summation of diagonal elements of a matrix. 
Finally, $[\mathbf{A}]_{i,:}$, $[\mathbf{A}]_{:,j}$, $[\mathbf{A}]_{i,j}$, and $[\mathbf{a}]_i$ denote the $i$-th row, the $j$-th column, the $(i,j)$-th element of matrix $\mathbf{A}$, and the $i$-th element of vector $\mathbf{a}$, respectively.

\section{Multi-Sector BD-RIS Model}

In this section, we propose the multi-sector BD-RIS-assisted communication model based on \cite{shen2021}, \cite{balanis2015}, and specify the channel model based on antenna theory \cite{pozar2011}. 

We consider a general wireless communication system consisting of an $N$-antenna transmitter, $K$ multi-antenna users, each of which has $N_{k}$ antennas, and an $M$-cell RIS.
Denote $\mathcal{N}=\{1,\ldots,N\}$, $\mathcal{K} = \{1,\ldots,K\}$, and $\mathcal{M} = \{1,\ldots, M\}$ as the set of indices of transmitting antennas, users, and cells, respectively. 
Each cell in this RIS consists of $L$ antennas, $L\ge 2$. Particularly, cell $m$ contains antennas belonging to the set $\mathcal{L}_m = \{m, M + m, \ldots, (L-1)M+m\}$, $\forall m \in \mathcal{M}$.
Therefore, the $M$-cell RIS is modeled as $LM$ antennas connected to an $ML$-port passive reconfigurable impedance network\footnote{The proposed RIS model in this work is purely passive without power consumption since the reconfigurable impedance network of the RIS consists of passive circuit components.}\cite{shen2021}. 
The channel between the transmitter and user $k$ without mismatching and mutual coupling among antennas is \cite{shen2021}
\begin{equation}
    \mathbf{H}_{k}=\mathbf{H}_{\mathrm{UT},k}+\overline{\mathbf{H}}_{\mathrm{UI},k}\mathbf{\Phi}\overline{\mathbf{H}}_\mathrm{IT},\forall k\in\mathcal{K},\label{eq:overall_channel}
\end{equation}
where $\mathbf{H}_{\mathrm{UT},k}\in\mathbb{C}^{N_{k}\times N}$ is the channel matrix for the transmitter-user link, $\overline{\mathbf{H}}_{\mathrm{UI},k}\in\mathbb{C}^{N_{k}\times LM}$, $\forall k\in\mathcal{K}$ and $\overline{\mathbf{H}}_\mathrm{IT}\in\mathbb{C}^{LM\times N}$ are channel matrices for the transmitter-RIS-user link.
$\mathbf{\Phi}\in\mathbb{C}^{LM\times LM}$ is the scattering
matrix of the $LM$-port reconfigurable impedance network satisfying 
\cite{pozar2011}
\begin{equation}
    \mathbf{\Phi}^{H}\mathbf{\Phi}=\mathbf{I}_{LM},
    \label{eq:general_RIS_constraint}
\end{equation}
when the $LM$-port reconfigurable impedance network is lossless. The scattering matrix $\mathbf{\Phi}$ is solely used to characterize the reconfigurable impedance network, which is independent of the characteristics of the antenna array. More details regarding the impact of mutual coupling from the antenna array can be found in \cite{shen2021}, \cite{pozar2011}.

The channel model (\ref{eq:overall_channel}) and the constraint for an $M$-cell RIS (\ref{eq:general_RIS_constraint}) are general expressions regardless of architectures of RIS, spatial arrangements and radiation pattern of RIS antenna elements. In the following, we will propose a novel multi-sector BD-RIS model and derive corresponding multi-sector BD-RIS-assisted channel model by specifying the spatial arrangement of RIS antenna elements and the architecture of the RIS. To this end, we first make the following two assumptions.

\begin{figure}
    \centering
    \includegraphics[width = 0.48\textwidth]{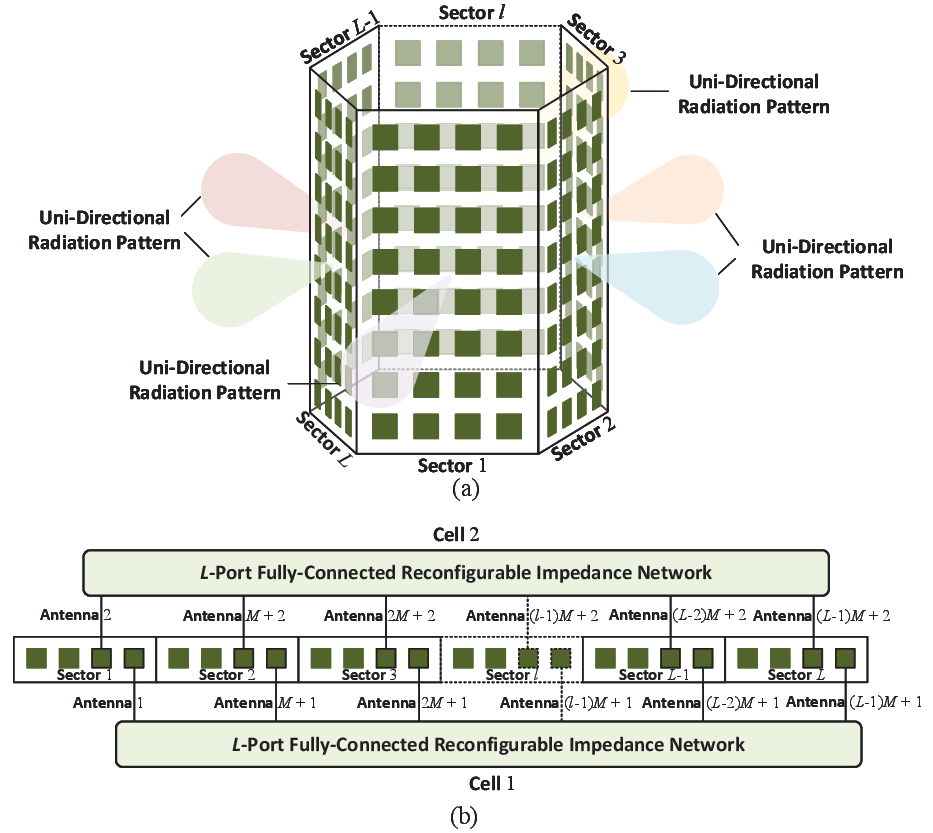}
    \caption{(a) A multi-sector BD-RIS with $LM$ uni-directional antennas modeled as an $L$-sector polygon. (b) Illustration of cells 1 and 2 taking the first row of the multi-sector BD-RIS (stretched in a straight line) as an example.}\label{fig:assumption} 
\end{figure}

\begin{assump}
    Each antenna of the RIS has a uni-directional radiation pattern with beamwidth $\frac{2\pi}{L}$.
\end{assump}

\begin{assump}
    The $L$ antennas in each cell are respectively placed at the side of the polygon such that each antenna covers $\frac{1}{L}$ space, as illustrated in Fig. \ref{fig:assumption} (a).
\end{assump}

Based on Assumptions 1 and 2, we divide the $M$-cell RIS into $L$ sectors. Specifically, sector $l$ covers $1/L$ space containing $M$ antennas, each of which belongs to $\mathcal{M}_l = \{(l-1)M+1, \ldots, lM\}$, $\forall l \in \mathcal{L} = \{1,\ldots,L\}$. The $M$-cell RIS is thus modeled as a polygon with $L$-sectors, which is referred to as multi-sector BD-RIS as shown in Fig. \ref{fig:assumption}(a). 
For clarity, in Fig. \ref{fig:assumption}(b), we also illustrate the ``cell'' of the multi-sector BD-RIS, where every $L$ antennas in the set $\mathcal{L}_m$ are connected to an $L$-port fully-connected reconfigurable impedance network to support the multi-sector mode.

\begin{figure}
    \centering
    \includegraphics[width = 0.48\textwidth]{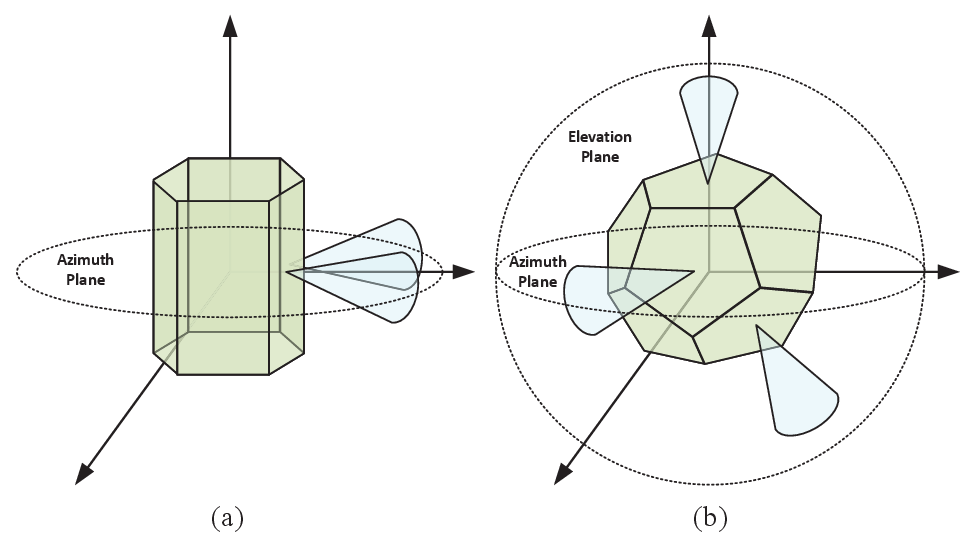}
    \caption{Multi-sector BD-RIS modeled as (a) a 2D polygon with 6 sectors and (b) a 3D polygon with 12 sectors.}\label{fig:3D_example} 
\end{figure}

\begin{remark}
    In this work, we focus on a 2D polygon for illustration and tractability. In Fig. \ref{fig:3D_example}(a), we provide a simple example when the multi-sector BD-RIS is modeled as a 2D polygon with 6 sectors. In this case, we assume the multi-sector BD-RIS has a 2D horizontal coverage at the azimuth plane. However, the concept is also applicable to 3D polygon where antennas are placed on a ``sphere''. For ease of understanding, we provide an example when the multi-sector BD-RIS is modeled as a 3D polygon with 12 sectors in Fig. \ref{fig:3D_example}(b). In this case, the multi-sector BD-RIS enables a 3D full-space coverage.
\end{remark}

With Assumptions 1 and 2, we can partition channel (\ref{eq:overall_channel}) as 
\begin{equation}
    \begin{aligned}
        \mathbf{H}_{k}= & \mathbf{H}_{\mathrm{UT},k}+\left[\overline{\mathbf{H}}_{\mathrm{UI},k,1}~\overline{\mathbf{H}}_{\mathrm{UI},k,2}~\ldots~\overline{\mathbf{H}}_{\mathrm{UI},k,L}\right]\\
        &~~~~\times\left[
        \begin{array}{cccc}
            \mathbf{\Phi}_{1,1} & \mathbf{\Phi}_{1,2} & \cdots & \mathbf{\Phi}_{1,L}\\
            \mathbf{\Phi}_{2,1} & \mathbf{\Phi}_{2,2} & \cdots & \mathbf{\Phi}_{2,L}\\
            \vdots & \vdots & \ddots & \vdots\\
            \mathbf{\Phi}_{L,1} & \mathbf{\Phi}_{L,2} & \cdots & \mathbf{\Phi}_{L,L}
        \end{array}\right]\left[
        \begin{array}{c}
            \overline{\mathbf{H}}_{\mathrm{IT},1}\\
            \overline{\mathbf{H}}_{\mathrm{IT},2}\\
            \vdots\\
            \overline{\mathbf{H}}_{\mathrm{IT},L}
        \end{array}\right]\\
        = & \mathbf{H}_{\mathrm{UT},k} + \sum_{l \in \mathcal{L}}\sum_{l'\in \mathcal{L}} \overline{\mathbf{H}}_{\mathrm{UI},k,l}\mathbf{\Phi}_{l,l'}\overline{\mathbf{H}}_{\mathrm{IT},l'}, \forall k \in \mathcal{K},
    \end{aligned}
    \label{eq:overall_channel1}
\end{equation}
where $\overline{\mathbf{H}}_{\mathrm{UI},k,l}=[\overline{\mathbf{H}}_{\mathrm{UI},k}]_{:,(l-1)M+1:lM}\in\mathbb{C}^{N_{k}\times M}$, and $\overline{\mathbf{H}}_{\mathrm{IT},l}=[\overline{\mathbf{H}}_\mathrm{IT}]_{(l-1)M+1:lM,:}\in\mathbb{C}^{M\times N}$, $\forall k\in\mathcal{K}$, 
$\forall l\in \mathcal{L}$ are channels between sector $l$ of the multi-sector BD-RIS
and user $k$, and between the transmitter and sector $l$, respectively. $\mathbf{\Phi}_{l,l'}=[\mathbf{\Phi}]_{(l-1)M+1:lM,(l'-1)M+1:l'M}\in\mathbb{C}^{M\times M}$,
$\forall l,l'\in\mathcal{L}$. 
To derive the multi-sector BD-RIS-assisted channel model, we further make the following assumption.

\begin{assump}
    The transmitter and $K_{1}$
    users with $\mathcal{K}_{1}=\{1,\ldots,K_{1}\}$, $0<K_{1}<K$
    are located within the coverage of sector 1 of the multi-sector BD-RIS; $K_l$
    users with $\mathcal{K}_{l}=\{\sum_{i \in \{1,\ldots,l-1\}}K_{i}+1,\ldots,\sum_{i \in \{1,\ldots,l\}}K_{i}\}$, $0<K_{l}<K$ are located within the coverage of sector $l$ of the multi-sector BD-RIS, $\forall l \in \mathcal{L}, l\ne 1$, as illustrated in Fig. \ref{fig:assumption_TV}. 
\end{assump}

\begin{figure}
    \centering
    \includegraphics[height=3.5 in]{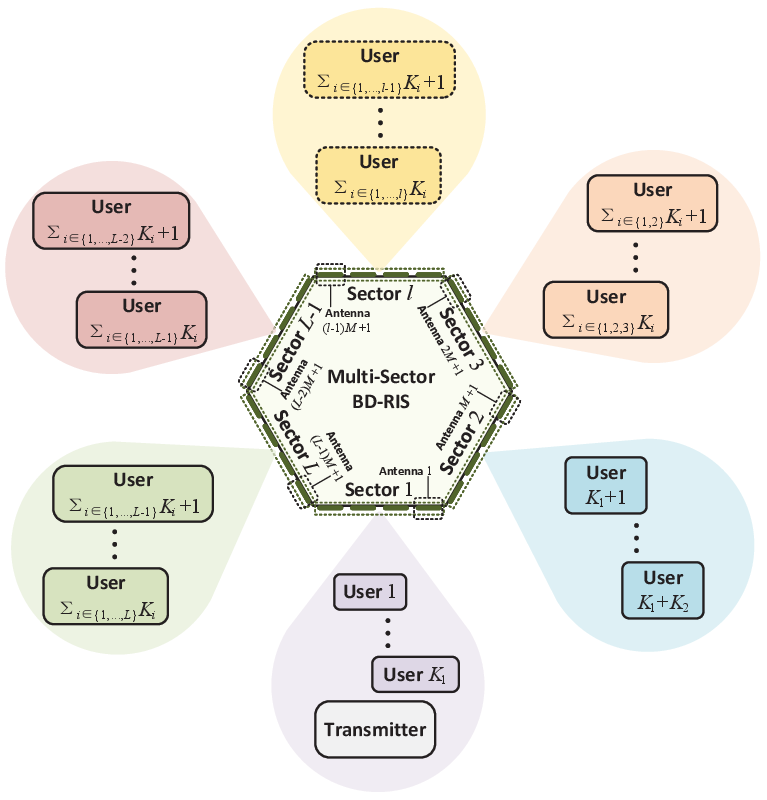}
    \caption{Top view for the locations of the $M$-cell multi-sector BD-RIS partitioning the whole space into $L$ sectors, the transmitter, and multiple users.}\label{fig:assumption_TV} 
\end{figure}

Based on Assumptions 1-3, we can deduce two corollaries.

\begin{corollary}
    The transmitter is out of the coverage of the uni-directional radiation pattern of sector $l$ of the multi-sector BD-RIS, $\forall l \in \mathcal{L}, l\ne 1$, that is $\overline{\mathbf{H}}_{\mathrm{IT},l}=\mathbf{0}$, $\forall l \in \mathcal{L}, l\ne 1$.
\end{corollary}

\begin{corollary}
    The users in $\mathcal{K}_{l}'$ are out of the coverage of the uni-directional radiation pattern of sector $l$ of the multi-sector BD-RIS, $\forall l\ne l'$, that is $\overline{\mathbf{H}}_{\mathrm{UI},k,l}=\mathbf{0}$, $\forall k\in\mathcal{K}_{l'}$, $\forall l\ne l', \forall l,l'\in\mathcal{L}$.
\end{corollary}

With Corollaries 1 and 2, we can simplify (\ref{eq:overall_channel1}) as 
\begin{equation}
    \mathbf{H}_{k}=  \mathbf{H}_{\mathrm{UT},k} + \overline{\mathbf{H}}_{\mathrm{UI},k,l}\mathbf{\Phi}_{l,1}\overline{\mathbf{H}}_{\mathrm{IT},1}, \forall k\in\mathcal{K}_{l}, \forall l \in \mathcal{L}.
    \label{eq:simplified_overall_channel}
\end{equation}
For simplicity, we use auxiliary notations $\mathbf{H}_{\mathrm{UI},k}=\overline{\mathbf{H}}_{\mathrm{UI},k,l}$, $\mathbf{\Phi}_l=\mathbf{\Phi}_{l,1}$, $\mathbf{H}_{\mathrm{IT}}=\overline{\mathbf{H}}_{\mathrm{IT},1}$, $\forall k\in\mathcal{K}_{l}$, $\forall l\in\mathcal{L}$. We also ignore the channel from the transmitter to the receiver ($\mathbf{H}_{\mathrm{UT},k} = \mathbf{0}$, $\forall k \in \mathcal{K}$). 
Then we can rewrite the simplified channel (\ref{eq:simplified_overall_channel})
as 
\begin{equation}
    \mathbf{H}_{k}= \mathbf{H}_{\mathrm{UI},k}\mathbf{\Phi}_l\mathbf{H}_\mathrm{IT}, \forall k\in\mathcal{K}_{l}, \forall l\in\mathcal{L}.
    \label{eq:simplified_overall_channel1}
\end{equation}
According to the general constraint (\ref{eq:general_RIS_constraint}), the constraint of matrices $\mathbf{\Phi}_l$, $\forall l\in \mathcal{L}$ of the multi-sector BD-RIS becomes
\begin{equation}
    \sum_{\forall l \in \mathcal{L}} \mathbf{\Phi}_l^H\mathbf{\Phi}_l = \mathbf{I}_M.
    \label{eq:ris_constraint}
\end{equation}

Now we establish the multi-sector BD-RIS-assisted communication model (\ref{eq:ris_constraint}) and (\ref{eq:simplified_overall_channel1}). In the following subsections, we will show the specific descriptions of matrices $\mathbf{\Phi}_l$, $\forall l \in \mathcal{L}$ when the cell-wise single-connected architecture is applied and channels $\mathbf{H}_k$, $\forall k \in \mathcal{K}$ including the impact of antenna gain.

\subsection{Cell-Wise Single-Connected Multi-Sector BD-RIS}

\begin{table}[t]
    \caption{Three Architectures of Multi-Sector BD-RIS}
    \centering 
    \begin{threeparttable}
        \begin{tabular}{|c|c|}
        \hline 
        (Cell-Wise) Architectures & Constraint\tabularnewline
        \hline 
        \multirow{2}*{Single-Connected} & \multirow{2}*{}{$\mathbf{\Phi}_l = \mathsf{diag}(\phi_{(l-1)M+1},\ldots,\phi_{lM})$}\\
        \multirow{2}*{} &\multirow{2}*{}{$\sum_{i\in \mathcal{L}_m} |\phi_i|^2 = 1$, $\forall m\in\mathcal{M}$} \tabularnewline
        \hline 
        \multirow{2}*{Group-Connected\tnote{$\dagger$}} & \multirow{2}*{}{$\mathbf{\Phi}_l = \mathsf{blkdiag}(\mathbf{\Phi}_{l,1}, \ldots, \mathbf{\Phi}_{l,G})$} \\
        \multirow{2}*{} & \multirow{2}*{}{$\sum_{\forall l \in \mathcal{L}}\mathbf{\Phi}_{l,g}^H\mathbf{\Phi}_{l,g} = \mathbf{I}_{\bar{M}}$,$\forall g\in\mathcal{G}$}\\
        \hline 
        Fully-Connected & $\sum_{\forall l \in \mathcal{L}}\mathbf{\Phi}_{l}^H\mathbf{\Phi}_{l} = \mathbf{I}_M$ \tabularnewline
        \hline 
        \end{tabular}
        \begin{tablenotes}
            \footnotesize
            \item[$\dagger$] $G$ is the number of groups \cite{li2022}. $\mathcal{G} = \{1,\ldots,G\}$, $\bar{M} = M/G$.
        \end{tablenotes}
    \end{threeparttable}
    \label{tab:architecture}
\end{table}

In this work, we assume the multi-sector BD-RIS has a cell-wise single-connected architecture, that is, different RIS cells are not connected to each other. 
Similar to \cite{li2022}, with different inter-cell circuit topologies, the multi-sector BD-RIS can also have different architectures, namely, cell-wise single/group/fully-connected architectures, which are summarized in Table \ref{tab:architecture}.
The cell-wise group/fully-connected architectures can provide higher performance gains thanks to the more general constrains as shown in Table \ref{tab:architecture}, but at the expense of higher circuit topology complexity. 
Here we focus only on cell-wise single-connected multi-sector BD-RIS and investigating cell-wise group/fully-connected architectures is left for future work. In the cell-wise single-connected case, $\mathbf{\Phi}_l$, $\forall l\in \mathcal{L}$ can be modeled as diagonal matrices:
\begin{equation}
    \begin{aligned}
    &\mathbf{\Phi}_l = \mathsf{diag}(\phi_{(l-1)M+1},\ldots,\phi_{lM}), \forall l \in \mathcal{L},
    \label{eq:ris_constraint_single}
    \end{aligned}
\end{equation}
and constraint (\ref{eq:ris_constraint}) becomes
\begin{equation}
    \sum_{i\in \mathcal{L}_m} |\phi_i|^2 = 1, \forall m \in \mathcal{M}.
    \label{eq:ris_constraint_single1}
\end{equation}

\begin{remark}
    We should clarify here that the proposed cell-wise single-connected multi-sector BD-RIS is an application of group-connected reconfigurable impedance network with group size equal to $L$.
    When antennas with different spatial arrangements are connected to the same reconfigurable impedance network, we have RIS models with different modes and architectures. In Fig. \ref{fig:example} we provide a simple example to facilitate understanding. We can observe from Fig. \ref{fig:example} that the circuit topology of the impedance network for a 2-cell 4-sector BD-RIS with cell-wise single-connected architecture is exactly the same as that for an 8-element group-connected RIS with group size equal to 4 \cite{shen2021}. The difference is that antennas within each cell of the multi-sector BD-RIS are placed at the side of the polygon, while antennas of the RIS in \cite{shen2021} are placed towards one direction. 
    Therefore, the freedom induced by the group-connected reconfigurable impedance network is utilized to support the multi-sector mode of the multi-sector BD-RIS, but to realize more flexible signal reflection for the group-connected RIS \cite{shen2021}.
\end{remark}

\begin{figure}
    \centering
    \includegraphics[width = 0.48\textwidth]{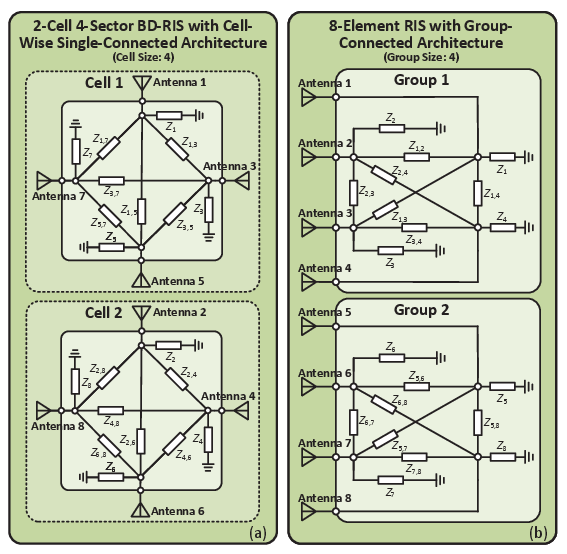}
    \caption{Examples of (a) cell-wise single-connected multi-sector BD-RIS and (b) group-connected RIS in \cite{shen2021}.}\label{fig:example} 
\end{figure}

\begin{remark}
    As per the circuit topology in Fig. \ref{fig:example}(a), to realize an $M$-cell $L$-sector BD-RIS with cell-wise single-connected architecture, $ML$ antennas and $M$ fully-connected $L$-port  reconfigurable impedance networks, each of which consists of $\frac{L(L+1)}{2}$ impedance components, are required. Therefore, the increasing performance with increasing $L$ is achieved at the expense of higher circuit complexity and cost. 
\end{remark}

We model $\phi_i$ as $\phi_i = \sqrt{\beta_i}e^{\jmath\vartheta_i}$, $\forall i \in \mathcal{L}_m$, $\forall m \in \mathcal{M}$, where $0\le\sqrt{\beta_i}\le 1$ and $0\le\vartheta_i\le 2\pi$ denote the amplitude and phase shift of $\phi_i$, respectively. 
Next, we will consider two cases of the multi-sector BD-RIS: 1) continuous case, that is both the amplitude and the phase shift have continuous values\footnote{To implement continuously controllable multi-sector BD-RIS, varactors could be used to realize the reconfigurable impedance network}; 2) discrete case, that is both the amplitude and the phase shift are chosen from finite codebooks\footnote{To implement discrete-value BD-RIS, PIN diodes can be used to realize the reconfigurable impedance network.}.

\subsubsection{Continuous Case} 
In this case, constraint (\ref{eq:ris_constraint_single1}) is transformed into
\begin{equation}
    \text{(\ref{eq:ris_constraint_single1})} \Rightarrow 
    \begin{cases}
        \phi_i = \sqrt{\beta_i}e^{\jmath\vartheta_i}, \forall i \in \mathcal{L}_m, \forall m \in \mathcal{M},\\
        0\le\beta_i\le 1, \forall i \in \mathcal{L}_m, \forall m \in \mathcal{M},\\
        \sum_{i\in \mathcal{L}_m} \beta_i = 1, \forall m \in \mathcal{M},\\
        0\le\vartheta_i< 2\pi, \forall i \in \mathcal{L}_m, \forall m \in \mathcal{M}.
    \end{cases}
    \label{eq:ris_constraint_continuous}
\end{equation}

\subsubsection{Discrete Case} We propose a simple quantization strategy when $\phi_i$ has discrete values. Specifically, the amplitude square $\beta_i$ is uniformly spaced within the range $[0,1]$, i.e., $\beta_i\in\mathcal{A} = \{\frac{a}{2^A-1}|a = 0,1,\ldots,2^A-1\}$ controlled by $A$ bits. The phase shift $\vartheta_i$ is uniformly spaced around the unit circle, i.e., $\vartheta_i\in\mathcal{B} = \{\frac{2\pi b}{2^B}|b = 0,1,\ldots,2^B-1\}$ controlled by $B$ bits. In this case, constraint (\ref{eq:ris_constraint_single1}) is transformed into
\begin{equation}
    \text{(\ref{eq:ris_constraint_single1})} \Rightarrow 
    \begin{cases}
        \phi_i = \sqrt{\beta_i}e^{\jmath\vartheta_i}, \forall i \in \mathcal{L}_m, \forall m \in \mathcal{M},\\
        \beta_i\in\mathcal{A}, \forall i \in \mathcal{L}_m, \forall m \in \mathcal{M},\\
        \sum_{i\in \mathcal{L}_m} \beta_i = 1, \forall m \in \mathcal{M},\\
        \vartheta_i\in \mathcal{B}, \forall i \in \mathcal{L}_m, \forall m \in \mathcal{M}.
    \end{cases}
    \label{eq:ris_constraint_discrete}
\end{equation}

\subsection{Channel Model}

We start by analyzing the path loss of the multi-sector BD-RIS-assisted communication system, which will be used to model the multi-sector BD-RIS-assisted channel (\ref{eq:simplified_overall_channel1}).

\begin{figure}
    \centering
    \includegraphics[width = 0.48\textwidth]{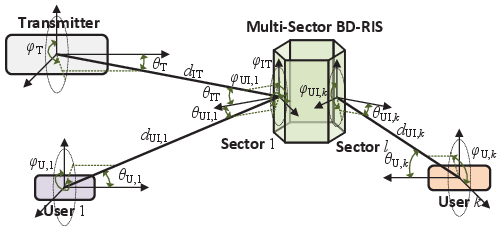}
    \caption{Relative locations among the multi-sector BD-RIS, the transmitter, and users.}\label{fig:locations} 
\end{figure}

The relative locations among the transmitter, the multi-sector BD-RIS, and users are illustrated in Fig. \ref{fig:locations}.
The power captured by each antenna in sector 1 of the multi-sector BD-RIS from the transmitting antenna is \cite{balanis2015}\footnote{In this paper, we consider a far-field case, that is the distance between the transmitter/users and the multi-sector BD-RIS is much larger than the size of the transmitter/RIS/users. In this case, we ignore the antenna spacing at the transmitter/RIS/users and the distance among different sectors of the multi-sector BD-RIS.} 
\begin{equation}
    P_{\mathrm{IT}} = \bar{P}_\mathrm{T}\left(\frac{\lambda}{4\pi}\right)^2d_\mathrm{IT}^{-\epsilon_\mathrm{IT}}
    G_\mathrm{T}(\theta_\mathrm{T},\varphi_\mathrm{T}) G_\mathrm{I}(\theta_{\mathrm{IT}},\varphi_\mathrm{IT}), 
    \label{eq:p_it}
\end{equation}
which is a generalization of the well-known free-space Friis transmission equation, where the power is allowed to decrease at a rate of $d_\mathrm{IT}^{-\epsilon_\mathrm{IT}}$ \cite{munoz2009position} with distance $d_\mathrm{IT}$ between the transmitter and multi-sector BD-RIS and the path loss exponent $\epsilon_\mathrm{IT}$. Specifically, in free-space propagations, the path loss exponent is 2, while in complicated propagations with obstacles, the path loss component is usually larger than 2.
In (\ref{eq:p_it}), $\bar{P}_\mathrm{T}$ is the power radiated from the transmitting antenna\footnote{We assume each transmitting antenna radiates the same power.}. $G_\mathrm{T}(\theta_\mathrm{T},\varphi_\mathrm{T})$ ($G_\mathrm{I}(\theta_\mathrm{IT},\varphi_\mathrm{IT})$) is the gain of each transmitting antenna (each RIS antenna)\footnote{We assume antennas at the same device (transmitter/RIS/user) have the same normalized power radiation pattern.}, where $\theta_\mathrm{T}$ ($\theta_\mathrm{IT}$) and $\varphi_\mathrm{T}$ ($\varphi_\mathrm{IT}$) are elevation and azimuth angles from the transmitter to sector 1 of the multi-sector BD-RIS (from sector 1 of the multi-sector BD-RIS to the transmitter), respectively. 
Then the power scattered from RIS antenna $m$, $\forall m \in \mathcal{M}_l$ in sector $l$, $\forall l \in \mathcal{L}$ can be represented as 
\begin{equation}
    P_{\mathrm{I},m} = P_\mathrm{IT}|\phi_{m}|^2 = P_\mathrm{IT}\beta_m, \forall m \in \mathcal{M}_l, \forall l \in \mathcal{L}.
    \label{eq:p_i}
\end{equation}  
Similarly, the power received by user $k$, $\forall k \in \mathcal{K}_l$ from RIS antenna $m$, $\forall m \in \mathcal{M}_l$ in sector $l$, $\forall l \in \mathcal{L}$ is given by
\begin{equation}
    \begin{aligned}
    P_{\mathrm{U},m,k} &= P_{\mathrm{I},m}\left(\frac{\lambda}{4\pi}\right)^2d_{\mathrm{UI},k}^{-\epsilon_{\mathrm{UI},k}} G_\mathrm{I}(\theta_{\mathrm{UI},k},\varphi_{\mathrm{UI},k})\\
    &\times G_\mathrm{U}(\theta_{\mathrm{U},k},\varphi_{\mathrm{U},k}),
    \forall m \in \mathcal{M}_l, \forall k \in \mathcal{K}_l, \forall l \in \mathcal{L},
    \label{eq:p_u}
    \end{aligned}
\end{equation}
where $\theta_{\mathrm{UI},k}$ ($\theta_{\mathrm{U},k}$) and $\varphi_{\mathrm{UI},k}$ ($\varphi_{\mathrm{U},k}$) are elevation and azimuth angles from RIS antenna $m$ to user $k$ (from user $k$ to RIS antenna $m$), respectively. $d_{\mathrm{UI},k}$ is the distance between the multi-sector BD-RIS and user $k$, and $\epsilon_{\mathrm{UI},k}$ is the path loss exponent.
Combining equations (\ref{eq:p_it})-(\ref{eq:p_u}), we can obtain
\begin{equation}
    \begin{aligned}
    P_{\mathrm{U},m,k} = \bar{P}_\mathrm{T}\beta_m\left(\frac{\lambda}{4\pi}\right)^4d_\mathrm{IT}^{-\epsilon_\mathrm{IT}}d_{\mathrm{UI},k}^{-\epsilon_{\mathrm{UI},k}}
    G_\mathrm{T}(\theta_\mathrm{T},\varphi_\mathrm{T}) &\\
    \times G_\mathrm{U}(\theta_{\mathrm{U},k},\varphi_{\mathrm{U},k})
    G_\mathrm{I}(\theta_{\mathrm{IT}},\varphi_\mathrm{IT}) G_\mathrm{I}(\theta_{\mathrm{UI},k},\varphi_{\mathrm{UI},k}),&\\
    \forall m \in \mathcal{M}_l, \forall k \in \mathcal{K}_l, \forall l \in \mathcal{L}.&
    \label{eq:p_u1}
    \end{aligned}
\end{equation}

To separate factors related to the transmitter/users and to the multi-sector BD-RIS, and focus only on the impact from the multi-sector BD-RIS, we assume 
antennas at the transmitter/users exhibit constant gain, that is $G_\mathrm{T}(\theta_\mathrm{T},\varphi_\mathrm{T}) = G_\mathrm{T}$ and $G_\mathrm{U}(\theta_{\mathrm{U},k},\varphi_{\mathrm{U},k}) = G_\mathrm{U}$, $\forall k \in \mathcal{K}_l$, $\forall l \in \mathcal{L}$.
Then with equation (\ref{eq:p_u1}), we define the cascaded path loss of the multi-sector BD-RIS-assisted communication system as 
\begin{equation}
    \begin{aligned}
        \zeta_k = \frac{(4\pi)^4}{\lambda^4G_\mathrm{T}G_\mathrm{U}}\times\underbrace{\frac{d_\mathrm{IT}^{\epsilon_\mathrm{IT}}d_{\mathrm{UI},k}^{\epsilon_{\mathrm{UI},k}}}{G_\mathrm{I}(\theta_{\mathrm{IT}},\varphi_\mathrm{IT}) G_\mathrm{I}(\theta_{\mathrm{UI},k},\varphi_{\mathrm{UI},k})}}_{\text{Vary with RIS antenna gain}}&, \\
        \forall k \in \mathcal{K}_l, \forall l \in \mathcal{L}&.
        \label{eq:path_loss}
        \end{aligned}
\end{equation}
It can be observed from equation (\ref{eq:path_loss}) that the path loss is highly associated with the gain of RIS antennas. 
In the following we will introduce two expressions of the gain of RIS antennas which are suitable for the multi-sector BD-RIS-assisted communication system.

\subsubsection{Idealized Radiation Pattern} 
We first consider an idealized radiation pattern function for each antenna of the $L$-sector BD-RIS, which is given by 
\begin{equation}
    F_\mathrm{I}^{\text{Idealized}}(\theta_\mathrm{I},\varphi_\mathrm{I}) = 
    \begin{cases}
        1, &0\le\theta_\mathrm{I}\le\frac{\pi}{L}, 0\le\varphi_\mathrm{I} \le 2\pi,\\
        0, &\frac{\pi}{L}<\theta_\mathrm{I} \le \pi, 0\le\varphi_\mathrm{I} \le 2\pi.
    \end{cases}
\end{equation}
According to antenna theory \cite{balanis2015}, by assuming 100\% radiation efficiency, we can calculate the corresponding antenna gain as 
\begin{equation}
    G_\mathrm{I}^{\text{Idealized}}(\theta_{\mathrm{I}},\varphi_\mathrm{I}) = \begin{cases}
        \frac{2}{1-\cos\frac{\pi}{L}}, &0\le\theta_\mathrm{I}\le\frac{\pi}{L},\\
        0, &\frac{\pi}{L}<\theta_\mathrm{I}\le\pi.
    \end{cases}
\end{equation}
With $0\le\theta_\mathrm{IT}\le\frac{\pi}{L}$, and $0\le\theta_{\mathrm{IU},k}\le\frac{\pi}{L}$, $\forall k \in \mathcal{K}_l$, $\forall l \in \mathcal{L}$, the corresponding path loss is given by
\begin{equation}
    \begin{aligned}
        \zeta_{k}^{\text{Idealized}} = \frac{4^3\pi^4d_\mathrm{IT}^{\epsilon_\mathrm{IT}}d_{\mathrm{UI},k}^{\epsilon_{\mathrm{UI},k}}(1-\cos\frac{\pi}{L})^2}{\lambda^4G_\mathrm{T}G_\mathrm{U}}, 
        \forall k \in \mathcal{K}_l, \forall l \in \mathcal{L}.
        \label{eq:path_loss_idealized}
    \end{aligned}
\end{equation}




\subsubsection{Practical Radiation Pattern} 
We then choose a proper and popular radiation pattern function for each antenna of the $L$-sector BD-RIS, which is given by \cite{balanis2015}
\begin{equation}
F_\mathrm{I}^{\text{Practical}}(\theta_\mathrm{I},\varphi_\mathrm{I}) = 
    \begin{cases}
    \cos^{\alpha_L}\theta_\mathrm{I}, &0\le\theta_\mathrm{I}\le\frac{\pi}{2}, 0\le\varphi_\mathrm{I} \le 2\pi,\\
    0, &\frac{\pi}{2}<\theta_\mathrm{I} \le \pi, 0\le\varphi_\mathrm{I} \le 2\pi,
    \end{cases}
\end{equation}
where $\alpha_L > 0$. 
Then we can obtain the half-power beamwidth (HPBW) of an $L$-sector BD-RIS by setting $F_\mathrm{I}(\theta_\mathrm{I},\varphi_\mathrm{I}) = 0.5$, which yields\footnote{Here we set the HPBW of each BD-RIS antenna as $\frac{2\pi}{L}$ to guarantee the full-space coverage. In this case, adjacent sectors are slightly overlapped with each other.} 
\begin{equation}
    \textrm{HPBW}_\mathrm{I}^{\text{Practical}} = 2\arccos\sqrt[\alpha_L]{0.5} = \frac{2\pi}{L},\label{eq:HPBW}
\end{equation}
where $\alpha_L$ can be determined by solving the equation function
$\sqrt[\alpha_L]{0.5} = \cos\frac{\pi}{L}$\footnote{For positive $\cos\frac{\pi}{L}$, the parameter $\alpha_L$ can also be calculated analytically as $\alpha_L = \frac{\log(0.5)}{\log(\cos\frac{\pi}{L})}$.}.
The antenna gain with $\eta = 1$ is  
\begin{equation}
    G_\mathrm{I}^{\text{Practical}}(\theta_{\mathrm{I}},\varphi_\mathrm{I}) = \begin{cases}
        2(\alpha_L+1)\cos^{\alpha_L}\theta_\mathrm{I}, &0\le\theta_\mathrm{I}\le\frac{\pi}{2},\\
        0, &\frac{\pi}{2}<\theta_\mathrm{I}\le\pi.
    \end{cases}
\end{equation}
With $0\le\theta_\mathrm{IT}\le\frac{\pi}{L}$, and $0\le\theta_{\mathrm{IU},k}\le\frac{\pi}{L}$, $\forall k \in \mathcal{K}_l$, $\forall l \in \mathcal{L}$, the path loss is given by
\begin{equation}
    \begin{aligned}
        \zeta_{k}^{\text{Practical}} = \frac{4^3\pi^4d_\mathrm{IT}^{\epsilon_\mathrm{IT}}d_{\mathrm{UI},k}^{\epsilon_{\mathrm{UI},k}}}{\lambda^4G_\mathrm{T}G_\mathrm{U}(\alpha_L+1)^2\cos^{\alpha_L}\theta_\mathrm{IT}\cos^{\alpha_L}\theta_{\mathrm{UI},k}}&, \\
        \forall k \in \mathcal{K}_l, \forall l \in \mathcal{L}&.
        \label{eq:path_loss_specific}
    \end{aligned}
\end{equation}

With the path loss model (\ref{eq:path_loss_idealized}) or (\ref{eq:path_loss_specific}), we can model the multi-sector BD-RIS-assisted channel $\mathbf{H}_k$ (\ref{eq:simplified_overall_channel1}) as a combination of large-scale and small-scale fading, i.e.,
\begin{equation}
    \mathbf{H}_k = \sqrt{\zeta_k^{-1}}\widetilde{\mathbf{H}}_{\mathrm{UI},k}\mathbf{\Phi}_l\widetilde{\mathbf{H}}_\mathrm{IT}, \forall k \in \mathcal{K}_l, \forall l \in \mathcal{L},
    \label{eq:channel_model}
\end{equation} 
where $\zeta_k$, $ \forall k \in \mathcal{K}_l$, $\forall l \in \mathcal{L}$ denote path loss components accounting for the large-scale fading. $\widetilde{\mathbf{H}}_o$, $\forall o\in\mathcal{O}\triangleq\{\mathrm{IT}\}\cup\{\mathrm{UI},k\}_{\forall k}$ is the small-scale fading, such as Rician fading, i.e.,
\begin{equation}
    \begin{aligned}
        \widetilde{\mathbf{H}}_o = \sqrt{\frac{\kappa_o}{\kappa_o + 1}}\widetilde{\mathbf{H}}_o^\mathrm{LoS} + \sqrt{\frac{1}{\kappa_o + 1}}\widetilde{\mathbf{H}}_o^\mathrm{NLoS}, ~ \forall o \in \mathcal{O},
    \end{aligned}\label{eq:rician_h}
\end{equation}
where $\kappa_o$, $\forall o \in \mathcal{O}$ denote the Rician factor, channels with superscript ``$\mathrm{LoS}$'' and ``$\mathrm{NLoS}$'' refer to line-of-sight (LoS) component modeled based on the steering vector and non-line-of-sight (NLoS) component modeled as i.i.d. Rayleigh fading, respectively. 
Specifically, the transmitter and receiver $k$ are modeled as an $N$-antenna uniform linear array (ULA) and an $N_k$-antenna ULA with half-wavelength antenna spacing, yielding the steering vectors 
\begin{subequations}
    \begin{align}
        \mathbf{a}_\mathrm{T}(\theta_\mathrm{T}) &= [1,e^{j\pi\sin\theta_\mathrm{T}},\ldots,e^{j\pi(N-1)\sin\theta_\mathrm{T}}]^T,\\
        \mathbf{a}_{\mathrm{U},k}(\theta_{\mathrm{U},k}) &= [1,e^{j\pi\sin\theta_{\mathrm{U},k}},\ldots,e^{j\pi(N_k-1)\sin\theta_{\mathrm{U},k}}]^T.
    \end{align}
\end{subequations}
Each sector of the multi-sector BD-RIS is modeled as an $M$-antenna uniform planer array (UPA) with $M_\mathrm{x}$ antennas in the x-direction and $M_\mathrm{y}$ antennas in the y-direction, yielding the steering vector
\begin{equation}
    \begin{aligned}
        \mathbf{a}_\mathrm{I}(\theta_o,\varphi_o) = [1,e^{j\pi\sin\theta_o\cos\varphi_o},\ldots,e^{j\pi(M_\mathrm{x}-1)\sin\theta_o\cos\varphi_o}]^T&\\
        \otimes[1,e^{j\pi\sin\theta_o\sin\varphi_o},\ldots,e^{j\pi(M_\mathrm{y}-1)\sin\theta_o\sin\varphi_o}]^T, \forall o\in\mathcal{O}&.
    \end{aligned}
\end{equation}
Therefore, the LoS components for the transmitter-RIS channel and for the RIS-user channels are respectively given by 
\begin{subequations}
    \begin{align}
    \widetilde{\mathbf{H}}_\mathrm{IT}^\mathrm{LoS} &= \mathbf{a}_\mathrm{I}(\theta_\mathrm{IT},\varphi_\mathrm{IT})\mathbf{a}_\mathrm{T}(\theta_\mathrm{T})^H,\\
    \widetilde{\mathbf{H}}_{\mathrm{UI},k}^\mathrm{LoS} &= \mathbf{a}_{\mathrm{U},k}(\theta_{\mathrm{U},k})\mathbf{a}_\mathrm{I}(\theta_{\mathrm{UI},k},\varphi_{\mathrm{UI},k})^H, \forall k\in\mathcal{K}.
    \end{align}\label{eq:los}
\end{subequations}

\begin{remark}
    There are threefold differences between STAR-RIS/IOS \cite{xu2021star}, \cite{zhang2022intelligent} and the proposed multi-sector BD-RIS. 1) STAR-RIS/IOS \cite{xu2021star}, \cite{zhang2022intelligent} is a special instance of our proposed multi-sector BD-RIS with $L = 2$. 2) Our model includes the effect of antenna beamwidth (and gain). Since the coverage of each sector of the multi-sector BD-RIS decreases with the growth of the number of sectors $L$, the antennas utilized to realize the multi-sector BD-RIS have narrower beamwidth, and thus higher gains. The antenna gain has a significant influence on the fading channels as analyzed in Section II-B and thus on the system performance as illustrated in Sections IV-VI. 
    3) The proposed model is built upon a group-connected reconfigurable impedance network \cite{shen2021}, \cite{balanis2015}, which cannot be done with IOS/STAR-RIS modeling \cite{xu2021star}, \cite{zhang2022intelligent} without a rigorous derivation follolwing network analysis theory. 
\end{remark}

Now with the channel model (\ref{eq:channel_model}), together with the path loss model (\ref{eq:path_loss_idealized}) or (\ref{eq:path_loss_specific}) and the small-scale fading model (\ref{eq:rician_h})-(\ref{eq:los}), and the cell-wise single-connected multi-sector BD-RIS model with continuous values (\ref{eq:ris_constraint_continuous}) or discrete values (\ref{eq:ris_constraint_discrete}), we finally establish the multi-sector BD-RIS-assisted communication model. 
In Sections III-VI, we will first discuss the potential applications of the multi-sector BD-RIS and then use this model to evaluate the performance of multi-sector BD-RIS-assisted communication systems.

\section{Potential Applications \\ of Multi-Sector BD-RIS}

In this section, we list potential applications of the proposed multi-sector BD-RIS.

\subsection{Millimeter Wave/Terahertz Communications and Sensing}
The need for more spectrum resources and higher frequency bands, such as millimeter wave or even higher Terahertz frequency bands, to cater to various emerging applications is one of important driving trends for 6G \cite{saad2019vision}, \cite{akyildiz2022terahertz}.  
Wireless communications and sensing on higher frequency bands usually suffer from higher spreading loss, have limited multi-path and high sparsity, and are vulnerable to blockages \cite{akyildiz2022terahertz}. 
The proposed multi-sector BD-RIS is particularly useful for addressing the above issues. 
First, since the antennas utilized to realize the multi-sector BD-RIS have higher gains with the increase of the number of sectors $L$, the multi-sector BD-RIS is favorable to compensate for the high spreading loss for millimeter wave/Terahertz communications and sensing.  
Second, the beams of the multi-sector BD-RIS can be highly directional (with increasing $L$) to apply to millimeter wave/Terahertz communication and sensing systems with strong LoS channels. 
Third, different from conventional RIS with only half-space coverage, the multi-sector BD-RIS enables full-space coverage, which is helpful for the coverage extension (especially when there are multiple obstacles). 
With the above benefits, multi-sector BD-RIS can be utilized in the recently popular and promising ISAC system to compensate for the severe path loss, boost wireless communication and sensing performance, while guaranteeing full-space coverage.

\vspace{0.3 cm}

\subsection{Cell-Free Networks}

Cell-free networks are considered as a promising technique to satisfy the seamless connectivity and demanding coverage requirements for 6G \cite{ngo2017cell}. 
Existing works have investigated the integration of conventional RISs (which can only reflect signals towards one side) and cell-free networks and shown the advantages of deploying RISs in cell-free networks \cite{le2021energy}, \cite{van2021ris}. 
However, in conventional RIS-aided networks, usually multiple RISs should be deployed since each RIS has limited coverage, leading to complicated topology design problems for practical RIS deployments.
The proposed multi-sector BD-RIS can simplify the deployment issue and provide additional benefits due to the following two reasons. 
First, since the multi-sector BD-RIS has full-space coverage, the locations of multi-sector BD-RISs can be more flexible compared to that of conventional ones.  
Specifically, when the number of sectors $L>2$, the multi-sector BD-RIS has a more adaptive structure, which means the multi-sector BD-RIS can be flexibly deployed into different locations, while STAR-RIS ($L=2$) usually needs to be attached to building surfaces. 
Second, with higher performance enhancement and wider coverage, the required number of multi-sector BD-RIS in cell-free networks can be much smaller, which will further facilitate the practical implementations.

\section{Scaling Law}

In this section, we analyze how the received power of the multi-sector BD-RIS-assisted system scales as a function of the number of sectors $L$ to show the benefit of the proposed multi-sector BD-RIS.
Different from the scaling law presented in \cite{wu2019intelligent}, where the impact of the number of RIS elements to the received power is investigated, here we focus on the impact of the number of sectors $L$, which is a novel factor based on the proposed multi-sector BD-RIS model, to the received power.

We consider a multi-sector BD-RIS-assisted SU-SISO system ($N = K = N_k = 1$) for simplicity. The single-antenna transmitter is allocated within the coverage of sector 1 of the multi-sector BD-RIS with $\theta_\mathrm{IT} = 0$. The single-antenna user is randomly allocated around the multi-sector BD-RIS, e.g., within the coverage of sector $l'$, $\forall l'\in\mathcal{L}$. 
The transmitting symbol is $s\in\mathbb{C}$ with $\mathbb{E}\{|s|^2\} = 1$. According to equation (\ref{eq:channel_model}), the received signal is expressed as 
\begin{equation}
    \begin{aligned}
        y = &\sqrt{P_\mathrm{T}}\mathbf{h}_\mathrm{UI}^H\mathbf{\Phi}_{l'}\mathbf{h}_\mathrm{IT}s + n\\
        = &\sqrt{P_\mathrm{T}\zeta^{-1}}\tilde{\mathbf{h}}_\mathrm{UI}^H\mathbf{\Phi}_{l'}\tilde{\mathbf{h}}_\mathrm{IT}s + n,
    \end{aligned}
\end{equation}
where $P_\mathrm{T}$ denotes the transmitting power, $\mathbf{h}_\mathrm{UI}\in\mathbb{C}^M$ and $\mathbf{h}_\mathrm{IT}\in\mathbb{C}^M$ denote channels from sector $l'$ of the multi-sector BD-RIS to the user and from the transmitter to sector 1 of the multi-sector BD-RIS with small-scale fading components $\tilde{\mathbf{h}}_\mathrm{UI}$ and $\tilde{\mathbf{h}}_\mathrm{IT}$, respectively. $\zeta$ is the path loss and $n$ is the noise. 
The received signal power at the user side is given by
\begin{equation}
    P_\mathrm{U} = \left |\sqrt{P_\mathrm{T}\zeta^{-1}}\tilde{\mathbf{h}}_\mathrm{UI}^H\mathbf{\Phi}_{l'}\tilde{\mathbf{h}}_\mathrm{IT}\right |^2 = P_\mathrm{T}\zeta^{-1}|\tilde{\mathbf{h}}_\mathrm{UI}^H\mathbf{\Phi}_{l'}\tilde{\mathbf{h}}_\mathrm{IT}|^2.
\end{equation}
To achieve maximum received signal power, we assume that when the user is within the coverage of sector $l'$, $\forall l' \in \mathcal{L}$, the corresponding sector is activated. In this case, we have 
\begin{equation}
    \mathbf{\Phi}_l = \begin{cases}
        \mathsf{diag}(\phi_{(l-1)M+1}, \ldots, \phi_{lM}), &l=l',\\
        \mathbf{0}, &l\ne l',
    \end{cases}\forall l \in \mathcal{L},
    \label{eq:ris_constraint_siso}
\end{equation}
with $|\phi_j| = 1, \forall j \in \mathcal{M}_{l'}$.
Then, when $\phi_{(l'-1)M+m} = -\angle\big([\tilde{\mathbf{h}}_\mathrm{UI}^H]_m [\tilde{\mathbf{h}}_\mathrm{IT}]_m\big)$, $\forall m \in \mathcal{M}$, the received power is maximized to 
\begin{equation}
    P_\mathrm{U}^{\max} = P_\mathrm{T}\zeta^{-1}\left(\sum_{\forall m\in\mathcal{M}}|[\tilde{\mathbf{h}}_\mathrm{UI}^H]_m [\tilde{\mathbf{h}}_\mathrm{IT}]_m|\right)^2.
\end{equation}
We assume $\tilde{\mathbf{h}}_\mathrm{UI}$ and $\tilde{\mathbf{h}}_\mathrm{IT}$ are LoS channels for simplicity, i.e., $\tilde{\mathbf{h}}_i = [e^{\jmath\psi_{i,1}}, \ldots, e^{\jmath\psi_{i,M}}]^T$, $\forall i \in \{\mathrm{UI, IT}\}$, while the following scaling laws can be easily generalized to different small-scale fading conditions. Then with $\phi_{(l'-1)M+m} = \psi_{\mathrm{UI},m} - \psi_{\mathrm{IT},m}$, $\forall m \in \mathcal{M}$, we have the following maximized received power
\begin{equation}
    P_\mathrm{U}^{\max} = P_\mathrm{T}\zeta^{-1}M^2.
    \label{eq:receive_power}
\end{equation}
It can be observed from equation (\ref{eq:receive_power}) that with fixed $P_\mathrm{T}$, the maximum received power is determined by the number of cells $M$ and the path loss, both of which are highly associated with the number of sectors $L$. In the following, we will derive the relationship between the maximum received power and $L$ by considering both idealized and practical radiation patterns as illustrated in Section II-B. 

\subsubsection{Idealized Radiation Pattern} In this case, the path loss is independent of the location of the user. Therefore, the maximum received signal power is given by 
\begin{equation}
    P_\mathrm{U}^{\text{Idealized}} = \frac{P_\mathrm{T}M^2\lambda^4G_\mathrm{T}G_\mathrm{U}}{4^3\pi^4d_\mathrm{IT}^2d_{\mathrm{UI}}^2(1-\cos\frac{\pi}{L})^2},
    \label{eq:receive_power_idealized}
\end{equation}
where $d_\mathrm{UI}$ denotes the distance between the multi-sector BD-RIS and the user and the path loss components are set as $\epsilon_\mathrm{IT} = \epsilon_\mathrm{UI} = 2$ for LoS propagations. 
From (\ref{eq:receive_power_idealized}) we can deduce that the maximum received power is proportional to the number of sectors $L (L\ge 2)$ when other terms in equation (\ref{eq:receive_power_idealized}) are fixed. However, when the number of antennas for the whole multi-sector BD-RIS, i.e., $ML$, is fixed, the number of antennas for each sector $M$ is inversely proportional to the number of sectors $L$. There might be a trade-off between $M$ and $L$, which will be discussed in the following simulations.



\subsubsection{Practical Radiation Pattern}
The maximum received signal power with $\theta_\mathrm{IT} = 0$ in this case is given by 
\begin{equation}
    P_\mathrm{U}^{\text{Practical}} = \frac{P_\mathrm{T}M^2\lambda^4G_\mathrm{T}G_\mathrm{U}}{4^3\pi^4d_\mathrm{IT}^2d_{\mathrm{IU}}^2}(\alpha_L+1)^2\cos^{\alpha_L}\theta_\mathrm{IU}.
    \label{eq:receive_power_practical}
\end{equation}
Given that the user is randomly allocated around the multi-sector BD-RIS, we have $\theta_\mathrm{IU}\sim \mathcal{U}(0,\pi/L)$.
Therefore, the average of $P_\mathrm{U}^{\text{Practical}}$ is given by 
\begin{equation}
    \begin{aligned}
    \overline{P}_\mathrm{U}^{\text{Practical}} &= \mathbb{E}\{P_\mathrm{U}^{\text{Practical}}\}\\
    &= \frac{P_\mathrm{T}M^2\lambda^4G_\mathrm{T}G_\mathrm{U}}{4^3\pi^4d_\mathrm{IT}^2d_{\mathrm{IU}}^2}(\alpha_L+1)^2\mathbb{E}\{\cos^{\alpha_L}\theta_\mathrm{IU}\}\\
    &\overset{(a)}{\le} \frac{P_\mathrm{T}M^2\lambda^4G_\mathrm{T}G_\mathrm{U}}{4^3\pi^4d_\mathrm{IT}^2d_{\mathrm{IU}}^2}(\alpha_L+1)^2\cos^{\alpha_L}\mathbb{E}\{\theta_\mathrm{IU}\}\\
    &= \frac{P_\mathrm{T}M^2\lambda^4G_\mathrm{T}G_\mathrm{U}}{4^3\pi^4d_\mathrm{IT}^2d_{\mathrm{IU}}^2}(\alpha_L+1)^2\cos^{\alpha_L}\frac{\pi}{2L}\\
    &= \tilde{P}_\mathrm{U}^{\text{Practical}},
    \end{aligned}
    \label{eq:receive_power_average}
\end{equation}
where (a) holds following the Jensen's inequality.
From the relationship $\sqrt[\alpha_L]{0.5} = \cos\frac{\pi}{L}$ we can obtain that the number of sectors $L$ is proportional to $\alpha_L$. 
In addition, with the growth of $L$, $\cos^{\alpha_L}\frac{\pi}{2L}$ is always within the range $0.5<\cos^{\alpha_L}\frac{\pi}{2L}<1$, while $\alpha_L$ grows significantly, such that $\tilde{P}_\mathrm{U}^{\text{Practical}}$ is proportional to the number of sectors $L$ (with fixed $M$). Similarly, when we fix the number of antennas $ML$ for the whole multi-sector BD-RIS, the relationship between the maximum received power and the number of sectors $L$ might not be that straightforward.

\begin{figure}
    \centering
    \subfigure[Number of antennas for each sector $M = 32$]{
    \label{fig:P_L_M_fixed}    
    \includegraphics[width = 0.47\textwidth]{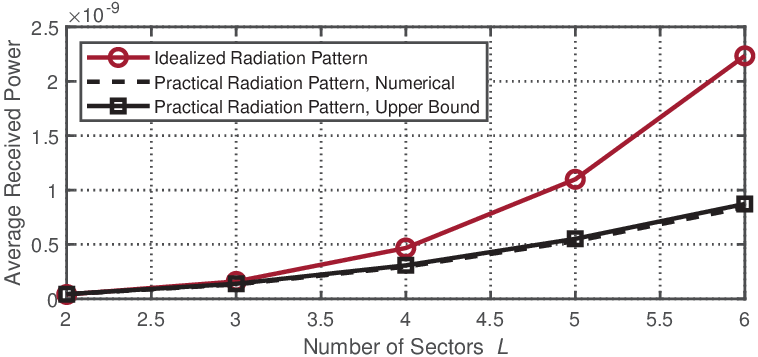}}
    \subfigure[Number of antennas for the whole RIS $ML = 180$]{
    \label{fig:P_L_Mtot_fixed}       
    \includegraphics[width = 0.47\textwidth]{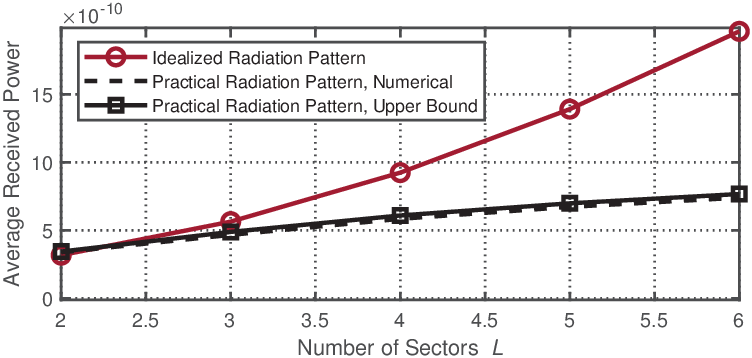}}
    \caption{Average received power versus the number of sectors $L$ for the multi-sector BD-RIS. ($G_\mathrm{T} = G_\mathrm{U} = 1$, $P_\mathrm{T} = 1$ W, $d_\mathrm{IT} = 100$ m, $d_\mathrm{IU} = 10$ m, $\lambda = c/f$ with $f = 2.4$ GHz.)}\label{fig:Power_L} 
\end{figure}

To clarify the relationship between the received power and the number of sectors $L$, in Fig. \ref{fig:P_L_M_fixed} we first plot the received power $P^{\text{Idealized}}_\mathrm{U}$ using an idealized radiation pattern, $\overline{P}^{\text{Practical}}_\mathrm{U}$ with a practical radiation pattern, and its upper bound $\tilde{P}^{\text{Practical}}_\mathrm{U}$ versus the number of sectors $L$ when the number of antennas for each sector is fixed to $M = 32$. It can be observed from Fig. \ref{fig:Power_L} that the received power increase with the growth of $L$, and that $\overline{P}^{\text{Practical}}_\mathrm{U}$ can finely approach its upper bound $\tilde{P}^{\text{Practical}}$. 
Then in Fig. \ref{fig:P_L_Mtot_fixed}, we plot the received power versus the number of sectors when the number of antennas for the whole multi-sector BD-RIS is fixed as $ML = 180$. In this case, the dimension of each sector of the multi-sector BD-RIS decreases with the growth of the number of sectors $L$. The received power still increases significantly, which indicates the impact of the antenna gain is more dominant than that of the dimension of RIS sectors. 
Based on the above analysis, the proposed multi-sector BD-RIS-assisted system is expected to achieve better performance when the number of sectors gets larger. To verify this conjecture, in the following section we will focus on the beamforming design to maximize the sum-rate performance for the multi-sector BD-RIS-assisted system.

\section{Beamforming Design for Multi-Sector BD-RIS-Assisted MU-MISO System}

In this section we apply the multi-sector BD-RIS in MU-MISO communication systems and propose efficient algorithms to design the multi-sector BD-RIS for both continuous and discrete cases. 

\subsection{Problem Formulation and Optimization Framework}

We consider a multi-sector BD-RIS-assisted MU-MISO system ($N_k = 1$, $\forall k \in \mathcal{K}$) in accordance with Section III.  
The transmitting symbol vector is $\mathbf{s} \triangleq [s_1, \ldots, s_K]^T \in \mathbb{C}^{K}$ with $\mathbb{E}\{\mathbf{s}\mathbf{s}^H\} = \mathbf{I}_{K}$. The transmitter applies a precoder matrix $\mathbf{W} \triangleq [\mathbf{w}_1, \ldots, \mathbf{w}_{K}] \in \mathbb{C}^{N\times K}$ with $\mathbf{w}_{k} \in \mathbb{C}^N$ for user $k$, $\forall k \in \mathcal{K}$.
With equations (\ref{eq:simplified_overall_channel1}) and (\ref{eq:ris_constraint_single}), the received signal $y_k$ for user $k$ is expressed as
\begin{equation}
    \label{eq:received_signal}
    \begin{aligned}
        y_k =
        &\mathbf{h}_{\mathrm{UI},k}^H\mathbf{\Phi}_l\mathbf{H}_\mathrm{IT}\mathbf{W}\mathbf{s} + n_{k},\\
        =&\mathbf{v}_{k,k}^H{\bm \phi}_ls_k + \sum_{\substack{p \in \mathcal{K}\\ p \ne k}}\mathbf{v}_{k,p}^H{\bm \phi}_ls_p
        + n_k, \forall k \in \mathcal{K}_l, \forall l \in \mathcal{L}, 
    \end{aligned}
\end{equation}
where $\mathbf{v}_{k,p} = (\mathbf{h}_{\mathrm{UI},k}^H\mathsf{diag}(\mathbf{H}_\mathrm{IT}\mathbf{w}_{p}))^H$, $\forall p,k \in \mathcal{K}$, $\bm{\phi}_l = [\phi_{(l-1)M+1},\ldots,\phi_{lM}]^T$, $\forall l\in\mathcal{L}$, $\mathbf{h}_{\mathrm{UI},k} \in \mathbb{C}^{M}$, $\forall k \in \mathcal{K}$ is the channel vector between the multi-sector BD-RIS and user $k$, and $n_{k} \sim \mathcal{CN}(0, \sigma_{k}^2)$ is the noise.
The sum-rate maximization problem can be formulated as 
\begin{subequations}
    \label{eq:problem0}
    \begin{align}
        \label{eq:obj0}
        \max_{\substack{\mathbf{W}\\ \{{\bm \phi}_l\}_{\forall l \in \mathcal{L}}}} &\sum_{l \in \mathcal{L}}\sum_{k\in \mathcal{K}_l}\log_2\left(1 + \frac{|\mathbf{v}_{k,k}^H{\bm \phi}_l|^2}{\sum_{\substack{p \in \mathcal{K}\\p\ne k}}|\mathbf{v}_{k,p}^H{\bm \phi}_l|^2 + \sigma_k^2}\right)\\
        \label{eq:p0_b}
        \mathrm{s.t.} ~~~ &\text{(\ref{eq:ris_constraint_continuous}) or (\ref{eq:ris_constraint_discrete})},\\
        \label{eq:p0_c}
        &\|\mathbf{W}\|_F^2 \le P_\mathrm{T}.
    \end{align}
\end{subequations}
Problem (\ref{eq:problem0}) is a typical sum-of-functions-of-ratio optimization, which can be solved by the general algorithm proposed in our previous work \cite{li2022}.
However, phase shifts and amplitudes of non-zero multi-sector BD-RIS elements are partitioned based on the illustration in Section III, which will facilitate the multi-sector BD-RIS design. 
Therefore, in this work we follow the framework in \cite{li2022}, whose main idea is to first transform the original problem into a multi-block optimization and then iteratively design each block until convergence, but reconsider the design of the multi-sector BD-RIS block.  
Specifically, we first apply $2K$ auxiliary variables, $\iota_k$ and $\tau_k$, $\forall k \in \mathcal{K}$, and transform (\ref{eq:problem0}) based on fractional programming theory \cite{FP2018} into the following four-block optimization:
\begin{equation}
    \label{eq:problem1}
    \begin{aligned}
        &\max_{\substack{\mathbf{W}, \{\iota_k\}_{\forall k \in \mathcal{K}}\\ \{\tau_k\}_{\forall k \in \mathcal{K}}, \{{\bm \phi}_l\}_{\forall l \in \mathcal{L}}}} \sum_{l \in \mathcal{L}}\sum_{k\in \mathcal{K}_l}\Big(\log_2(1 + \iota_k) - \iota_k + 2\sqrt{1 + \iota_{k}}\\
        &~~~~~~~~~~~~~~~~\times \Re\{\tau_{k}^*\mathbf{v}_{k,k}^H{\bm \phi}_l\} - |\tau_k|^2\big(\sum_{p \in \mathcal{K}}|\mathbf{v}_{k,p}^H{\bm \phi}_l|^2+\sigma_k^2\big)\Big)\\
        &~~~~~~~~~\mathrm{s.t.} ~~~~~  \textrm{(\ref{eq:p0_b}), (\ref{eq:p0_c})}.
    \end{aligned}
\end{equation}
Solutions to blocks $\{\iota_k\}_{\forall k \in \mathcal{K}}$, $\{\tau_k\}_{\forall k \in \mathcal{K}}$, and $\mathbf{W}$ are similar to that in \cite{li2022} with simple scaling of matrix dimensions, which are omitted due to the space limitation. 
In the following subsection, we will propose efficient algorithms to determine block $\{{\bm \phi}_l\}_{\forall l \in \mathcal{L}}$ for both continuous and discrete cases. 

\subsection{Solution to $\{{\bm \phi}_l\}_{\forall l \in \mathcal{L}}$}

When blocks $\{\iota_k\}_{\forall k \in \mathcal{K}}$, $\{\tau_k\}_{\forall k \in \mathcal{K}}$, and $\mathbf{W}$ are given, we can remove the constant terms in problem (\ref{eq:problem1}) and formulate the sub-problem with respect to $\{{\bm \phi}_l\}_{\forall l \in \mathcal{L}}$ as 
\begin{equation}
    \label{eq:problem_phi}
    \begin{aligned}     
        &\max_{\{{\bm \phi}_l\}_{\forall l \in \mathcal{L}}} \sum_{l \in \mathcal{L}}\Big(2\Re\{\tilde{\mathbf{v}}_l^H{\bm \phi}_l\} - {\bm \phi}_l^H\mathbf{V}_l{\bm \phi}_l\Big)\\
        &~~~\mathrm{s.t.} ~~~~ \textrm{(\ref{eq:p0_b})},
    \end{aligned}
\end{equation}
where $\tilde{\mathbf{v}}_l = (\sum_{k\in \mathcal{K}_l}\sqrt{1 + \iota_{k}}\tau_{k}^*\mathbf{v}_{k,k}^H)^H$, and $\mathbf{V}_l = \sum_{k\in \mathcal{K}_l}|\tau_k|^2\sum_{p \in \mathcal{K}}\mathbf{v}_{k,p}\mathbf{v}_{k,p}^H$, $\forall l \in \mathcal{L}$.
We successively design each cell of the multi-sector BD-RIS with fixed other cells until the convergence is guaranteed. Specifically, the sub-problem for the design of cell $m$, i.e., $\phi_i$, $\forall i \in \mathcal{L}_m$ while fixing the rest of $M-1$ cells is given by 
\begin{subequations}
    \label{eq:sub_phi_cell_m}
    \begin{align}
        &\min_{\{\phi_i\}_{\forall i \in \mathcal{L}_m}} \sum_{i \in \mathcal{L}_m}\Big(\nu_i|\phi_i|^2 
        + 2\Re\{\phi_{i}^*\chi_i\}\Big)\\
        \Leftrightarrow  &\min_{\substack{\{\beta_i\}_{\forall i \in \mathcal{L}_m} \\ 
        \label{eq:obj_phi_cell_m}
        \{\vartheta_i\}_{\forall i \in \mathcal{L}_m}}} \sum_{i \in \mathcal{L}_m}\Big(\nu_i\beta_i + 2|\chi_i|\sqrt{\beta_i}\cos({\angle\chi_i - \vartheta_i})\Big)\\
        &~~~~~\mathrm{s.t.} ~~~~ \textrm{(\ref{eq:p0_b})},
    \end{align}
\end{subequations}
where $\chi_{(l-1)M+m} = \sum_{n\ne m}[\mathbf{V}_l]_{m,n}\phi_{(l-1)M+n} - [\tilde{\mathbf{v}}_l]_m$, and $\nu_{(l-1)M+m} = [\mathbf{V}_l]_{m,m}$, $\forall l \in \mathcal{L}$, $\forall m \in \mathcal{M}$.
In the following we will consider the phase shift and amplitude design for both continuous and discrete cases.

\subsubsection{Continuous Case with Constraint (\ref{eq:ris_constraint_continuous})}
In this case, phase shifts $\vartheta_i$, $\forall i \in \mathcal{L}_m$ can be easily obtained by setting $\cos({\angle\chi_i - \vartheta_i}) = -1$, $\forall i \in \mathcal{L}_m$, which yields 
\begin{equation}\label{eq:theta_continuous}
    \begin{aligned}
        \theta_i^\star = \angle\chi_i \pm \pi, \forall i \in \mathcal{L}_m.
    \end{aligned}
\end{equation}
Substituting $\theta_{i}^\star$, $\forall i \in \mathcal{L}_m$ into objective (\ref{eq:obj_phi_cell_m}), we can obtain a real-value problem with respect to $\beta_i$, $\forall i \in \mathcal{L}_m$ as 
\begin{subequations}
    \label{eq:amplitude_cell_m_continuous}
    \begin{align}
        \min_{\{\beta_i\}_{\forall i \in \mathcal{L}_m}} &\sum_{i \in \mathcal{L}_m}\Big(\nu_i\beta_i - 2|\chi_i|\sqrt{\beta_i}\Big)\\
        \Leftrightarrow ~~\min_{{\bm \beta}_m} ~\;\;&\underbrace{{\bm \beta}_m^T\mathbf{\Psi}_m{\bm \beta}_m - 2{\bm \chi}_m^T{\bm \beta}_m}_{=f({\bm \beta}_m)}\\
        \label{eq:amplitude_constraint_c}
        \mathrm{s.t.} ~~~ &0\le\beta_i\le 1, \forall i \in \mathcal{L}_m,\\  
        \label{eq:amplitude_constraint_d}  
        &{\bm \beta}_m^T{\bm \beta}_m = 1,
    \end{align}
\end{subequations}
where ${\bm \beta}_m = [\sqrt{\beta_{m}}, \sqrt{\beta_{M + m}}, \ldots, \sqrt{\beta_{(L-1)M+m}}]^T\in\mathbb{R}^L$, $\mathbf{\Psi}_m = \mathsf{diag}(\nu_{m}, \nu_{M+m}, \ldots, \nu_{(L-1)M+m})\in\mathbb{R}^{L\times L}$, and ${\bm \chi}_m = [|\chi_{m}|, |\chi_{M+m}|, \ldots, |\chi_{(L-1)M+m}|]^T\in\mathbb{R}^L$.
Since $|\chi_i| > 0$, $\forall i \in \mathcal{L}_m$, constraint (\ref{eq:amplitude_constraint_c}) can always be satisfied when problem (\ref{eq:amplitude_cell_m_continuous}) achieve its minimum. This observation motivates us to ignore constraint (\ref{eq:amplitude_constraint_c}), which yields 
\begin{equation}
    \label{eq:amplitude_cell_m_continuous1}
    {\bm \beta}_m^\star = \arg\min_{{\bm \beta}_m^T{\bm \beta}_m = 1} f({\bm \beta}_m)
\end{equation}
The optimal solution to problem (\ref{eq:amplitude_cell_m_continuous1}) can be determined by the following lemma \cite{hager2001}. 

\begin{lemma}
${\bm \beta}_m$ is the optimal solution to problem (\ref{eq:amplitude_cell_m_continuous1}) if and only if there exists $\mu \in \mathbb{R}$ such that $(\mathbf{\Psi}_m + \mu\mathbf{I}_L){\bm \beta}_m = {\bm \chi}_m$, $\mathbf{\Psi}_m + \mu\mathbf{I}_L$ positive semi-definite, and ${\bm \beta}_m^T{\bm \beta}_m = 1$.
\end{lemma}

\begin{proof}
Please refer to Appendix A. 
\end{proof}

According to Lemma 1, ${\bm \beta}_m$ has a closed-form solution as a function of $\mu$, which is constrained by 
\begin{equation}
    \begin{aligned}
        g(\mu) &= {\bm \beta}_m^{\star T}{\bm \beta}_m^{\star}
        = {\bm \chi}_m^T(\mathbf{\Psi}_m + \mu\mathbf{I}_L)^{-2}{\bm \chi}_m\\
        &= \mathsf{Tr}\{(\mathbf{\Psi}_m + \mu\mathbf{I}_L)^{-2}{\bm \chi}_m{\bm \chi}_m^T\}\\
        &\overset{\text{(a)}}{=} \mathsf{Tr}\{(\mathbf{U}_m\mathbf{\Sigma}_m\mathbf{U}_m^T + \mu\mathbf{I}_L)^{-2}{\bm \chi}_m{\bm \chi}_m^T\}\\
        &= \mathsf{Tr}\{(\mathbf{\Sigma}_m + \mu\mathbf{I}_L)^{-2}\mathbf{U}_m^T{\bm \chi}_m{\bm \chi}_m^T\mathbf{U}_m\}\\
        &\overset{\text{(b)}}{=} \sum_{l\in \mathcal{L}}\frac{\tilde{\chi}_{m,l}^2}{(\varepsilon_{m,l} + \mu)^2} = 1,  
    \end{aligned}
\end{equation} 
where (a) holds by performing the eigenvalue decomposition $\mathbf{\Psi}_m = \mathbf{U}_m\mathbf{\Sigma}_m\mathbf{U}_m^T$ with $\mathbf{\Sigma}_m = \mathsf{diag}(\varepsilon_{m,1}, \ldots, \varepsilon_{m,L})$, $\varepsilon_{m,1}\le\varepsilon_{m,2}\le\ldots\le\varepsilon_{m,L}$, and $\mathbf{U}_m = [\mathbf{u}_{m,1}, \ldots, \mathbf{u}_{m,L}]$; (b) holds by defining $\tilde{\chi}_{m,l} = {\bm \chi}_m^T\mathbf{u}_{m,l}$, $\forall l \in \mathcal{L}$. Recall that $|\chi|_i > 0$, $\nu_i > 0$, $\forall i \in \mathcal{L}_m$. We have $\tilde{\chi}_{m,l} \ne 0$ and thus $\varepsilon_{m,l} + \mu \ne 0$. In addition, according to Lemma 1, $\mathbf{\Psi}_m + \mu\mathbf{I}$ is positive semi-definite such that $\varepsilon_{m,l} + \mu \ge 0$. Therefore, the optimal value of $\mu$ should be chosen such that $\varepsilon_{m,l} + \mu > 0$, $\forall l \in \mathcal{L}$ and $g(\mu) = 1$.    
In the following lemma, we derive the upper bound and lower bound of $\mu$ \cite{hager2001}.

\begin{lemma}
Define $\mathcal{E} = \{l |\varepsilon_{m,l} = \varepsilon_{m,1}, \forall l \in \mathcal{L}\}$. The optimal solution of $\mu$ is chosen within the range 
$\sqrt{\sum_{\forall l \in \mathcal{E}}\tilde{\chi}_{m,l}^2} - \varepsilon_{m,1} \le\mu\le \|{\bm \chi}_m\|_2 - \varepsilon_{m,1}$.
\end{lemma}

\begin{proof}
Please refer to Appendix B. 
\end{proof}

After determining the upper and lower bounds of $\mu$ according to Lemma 2, we can use efficient methods, such as bisection search, to find the optimal $\mu^\star$ such that $g(\mu^\star) = 1$. The optimal solution of problem (\ref{eq:amplitude_cell_m_continuous1}) is given by 
\begin{equation}
    \label{eq:opt_beta_continuous}
    {\bm \beta}_m^\star = (\mathbf{\Psi}_m + \mu^\star\mathbf{I}_L)^{-1}{\bm \chi}_m.
\end{equation}

Combining (\ref{eq:theta_continuous}) and (\ref{eq:opt_beta_continuous}), we finally obtain the solution of problem (\ref{eq:sub_phi_cell_m}) as 
$\phi_i^\star = \sqrt{\beta_i^\star}e^{\jmath\vartheta_i^\star}$, $\forall i \in \mathcal{L}_m$. 

\subsubsection{Discrete Case with Constraint (\ref{eq:ris_constraint_discrete})}
In this case, 
we first find the solution of phase shifts $\vartheta_i^\star$, $\forall i \in \mathcal{L}_m$ for continuous case by (\ref{eq:theta_continuous}). Then we apply a simple quantization operation to obtain the discrete phase shifts as 
\begin{equation}
    \label{eq:theta_discrete}
    \vartheta_i^\sharp = \Big\lceil \frac{\vartheta_i^\star}{\triangle_B} \Big\rfloor \times \triangle_B, \forall i \in \mathcal{L}_m,
\end{equation}   
where $\triangle_B = \frac{2\pi}{2^B}$ denotes the resolution of phase shift and $\lceil\cdot\rfloor$ denotes the rounding operation. Then we substitute (\ref{eq:theta_discrete}) to problem (\ref{eq:sub_phi_cell_m}), which yields
\begin{equation}
    \label{eq:aplitude_cell_m_discrete}
    \begin{aligned}
        \beta_i^\sharp = \arg\min_{\substack{\sum_{i\in\mathcal{L}_m}\beta_i = 1 \\ \beta_i\in \mathcal{A}}} ~\sum_{\forall i \in \mathcal{L}_m}f_i(\beta_i)
    \end{aligned}
\end{equation}
where $f_i(\beta_i) = \beta_i\nu_i + 2|\chi_i|\sqrt{\beta_i}\cos(\angle\chi_i - \vartheta_i^\sharp)$, $\forall i \in \mathcal{L}_m$.
Ignoring the discrete constraint and following the procedure for the design of continuous case, we can obtain the continuous solution $\beta_i^\star$, $\forall i \in \mathcal{L}_m$, which is further quantized by 
\begin{equation}
    \label{eq:amplitude_discrete}
    \beta_i^\sharp = \Big\lceil \frac{\beta_i^\star}{\triangle_A} \Big\rfloor \times \triangle_A, \forall i \in \mathcal{L}_m,
\end{equation} 
where $\triangle_A = \frac{1}{2^A-1}$ denotes the resolution of amplitude square.
It is worth noting that there is no guarantee of $\sum_{i\in\mathcal{L}_m}\beta_i^\sharp = 1$. More specifically, we have $\sum_{i\in\mathcal{L}_m}\beta_i^\sharp = 1 + \varsigma\triangle_A$, $\zeta \in \mathbb{Z}$ due to the rounding operation. In the following, we will propose a simple yet efficient algorithm to adjust the quantization result (\ref{eq:amplitude_discrete}) so that the constraint is satisfied. The proposed algorithm is described in Phases 1-3.

\textit{Phase 1:} Let $\varsigma = \frac{1-\sum_{\forall i \in \mathcal{L}_m}\beta_i^{\sharp}}{\triangle_A}$. If $\varsigma = 0$, the procedure stops; if $\varsigma > 0$, go to Phase 2; if $\varsigma < 0$, go to Phase 3.

\textit{Phase 2:} Set $\varsigma' = 1$. In the $\varsigma'$-th loop, do the following procedures until $\varsigma' > \varsigma$: For $\forall i \in \mathcal{L}_m$, if $\beta_i^\sharp < 1$, calculate $\delta_i = f_i(\beta_i^\sharp + \triangle_A) - f_i(\beta_i^\sharp)$; otherwise, set $\delta_i$ to a sufficient large number, e.g., $\infty$. Then find $i'$ such that $\delta_{i'} = \min_{\forall i \in \mathcal{L}_m}\delta_i$ and replace $\beta_{i'}^\sharp$ with $\beta_{i'}^\sharp + \triangle_A$. Update $\varsigma' = \varsigma' + 1$.

\textit{Phase 3:} Set $\varsigma' = 1$. In the $\varsigma'$-th loop, do the following procedures until $\varsigma' > -\varsigma$: For $\forall i \in \mathcal{L}_m$, if $\beta_i^\sharp > 0$, calculate $\delta_i = f_i(\beta_i^\sharp - \triangle_A) - f_i(\beta_i^\sharp)$; otherwise, set $\delta_i$ to $\infty$. Then find $i'$ satisfying $\delta_{i'} = \min_{\forall i \in \mathcal{L}_m}\delta_i$ and replace $\beta_{i'}^\sharp$ with $\beta_{i'}^\sharp - \triangle_A$. Update $\varsigma' = \varsigma' + 1$.

After processing Phases 1-3, we can obtain the modified $\beta_i^\sharp$, $\forall i \in \mathcal{L}_m$ satisfying the constraint $\sum_{\forall i \in \mathcal{L}_m}\beta_i^\sharp = 1$, and the solution of problem (\ref{eq:sub_phi_cell_m}) as $\phi_i^\sharp = \sqrt{\beta_i^\sharp}e^{\jmath\vartheta_i^\sharp}$, $\forall i \in \mathcal{L}_m$.

\subsubsection{Summary} 
With solutions to sub-problem (\ref{eq:sub_phi_cell_m}), the procedures of multi-sector BD-RIS design for continuous/discrete cases are summarized as following Steps 1-4:

\textit{Step 1:} Calculate $\tilde{\mathbf{v}}_l = (\sum_{k\in \mathcal{K}_l}\sqrt{1 + \iota_{k}}\tau_{k}^*\mathbf{v}_{k,k}^H)^H$ and $\mathbf{V}_l = \sum_{k\in \mathcal{K}_l}|\tau_k|^2\sum_{p \in \mathcal{K}}\mathbf{v}_{k,p}\mathbf{v}_{k,p}^H$, $\forall l \in \mathcal{L}$. Go to Step 2 for continuous case or Step 3 for discrete case. 

\textit{Step 2:} For continuous case, do the following procedures for $\forall m \in \mathcal{M}$: 
1) Calculate $\chi_i$ in problem (\ref{eq:sub_phi_cell_m}) and obtain $\theta_i^\star = \angle\chi_i \pm \pi$, $\forall i \in \mathcal{L}_m$. 
2) Find $\mu_m^\star$ by bisection search based on Lemma 2 such that $g(\mu_m^\star) = 1$, and obtain ${\bm \beta}_m^\star = (\mathbf{\Psi}_m + \mu_m^\star\mathbf{I}_L)^{-1}{\bm \chi}_m$. 
3) Update $\phi_i^\star = \sqrt{\beta_i^\star}e^{\jmath\vartheta_i^\star}$, $\forall i \in \mathcal{L}_m$.

\textit{Step 3:} For discrete case, do the following procedures for $\forall m \in \mathcal{M}$: 
1) Calculate $\chi_i$ in problem (\ref{eq:sub_phi_cell_m}), obtain $\theta_i^\star$, $\forall i \in \mathcal{L}_m$ as in Step 2 and quantize $\theta_i^\star$ by $\vartheta_i^\sharp = \lceil \frac{\vartheta_i^\star}{\triangle_B} \rfloor \times \triangle_B$, $\forall i \in \mathcal{L}_m$. 
2) Obtain the continuous solution ${\bm \beta}_m^\star$ of problem (\ref{eq:amplitude_discrete}) as in Step 2 and quantize ${\bm \beta}_m^\star$ by $\beta_i^\sharp = \lceil \frac{\beta_i^\star}{\triangle_A} \rfloor \times \triangle_A$, $\forall i \in \mathcal{L}_m$. 
3) Execute Phases 1-3.
4) Update $\phi_i^\sharp = \sqrt{\beta_i^\sharp}e^{\jmath\vartheta_i^\sharp}$, $\forall i \in \mathcal{L}_m$.

\textit{Step 4:} Repeat Step 2 for continuous case or Step 3 for discrete case until the convergence of vectors ${\bm \phi}_l^\star$, $\forall l \in \mathcal{L}$ or ${\bm \phi}_l^\sharp$, $\forall l \in \mathcal{L}$ is achieved.

\section{Performance Evaluation}

In this section, we evaluate the performance of the multi-sector BD-RIS-assisted MU-MISO system. 
Channels from the transmitter to the RIS and from the multi-sector BD-RIS to users are modeled as the combination of the large-scale and the small-scale Rician fading as illustrated in Sec. II-B.
Antenna gains for transmitter antennas and user antennas are fixed to $G_\mathrm{T} = G_\mathrm{U} = 1$. 
The frequency of the transmit signal is set as $f = 2.4$ GHz. 
Distances between the transmitter and the multi-sector BD-RIS and between the multi-sector BD-RIS and users are $d_\mathrm{TI} = 100$ m and $d_{\mathrm{IU},k} = 10$ m, $\forall k \in \mathcal{K}$, respectively. 
The path loss components are set as $\epsilon_o = 2.5$ when the Rician factors are set as $\kappa_o = 0$ dB, and $\epsilon_o = 2$ when the Rician factors are set as $\kappa_o = 10$ dB, $\forall o\in\mathcal{O}$.
The noise power is $\sigma_{k}^2 = -80$ dBm, $\forall k \in \mathcal{K}$.
In the following simulations, we assume $K$ users are uniformly divided into $L$ groups, each of which contains $K/L$ users randomly located within the coverage of one sector.  
We also fix the total number of multi-sector BD-RIS antennas $ML$ for fair comparison, i.e., the antenna gain of each antenna increases with increasing $L$, while the dimension of each sector $M$ decreases with $L$.

\begin{figure}
    \centering
    \subfigure[$\kappa_o = 0$ dB, $\forall o\in\mathcal{O}$, $P_\mathrm{T} = 20$ dBm]{
    \label{fig:SR_Mtot_fixed_con_0}
    \includegraphics[width = 0.48\textwidth]{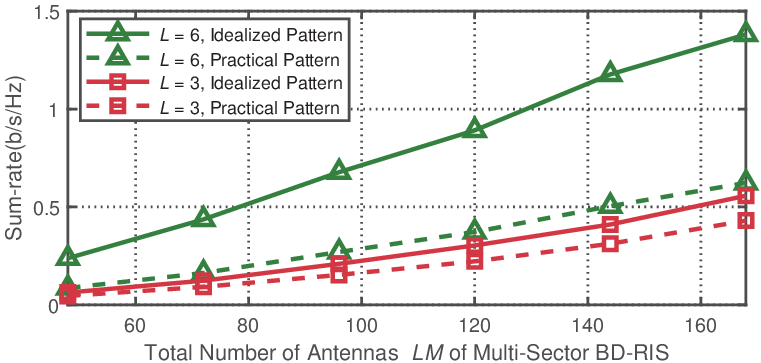}}
    \subfigure[$\kappa_o = 10$ dB, $\forall o\in\mathcal{O}$, $P_\mathrm{T} = 10$ dBm]{
    \label{fig:SR_Mtot_fixed_con_10}
    \includegraphics[width = 0.48\textwidth]{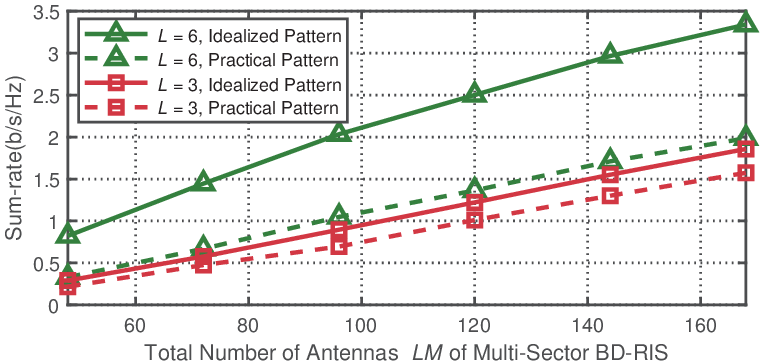}}
    \caption{Sum-rate versus total number of antennas of multi-sector BD-RIS ($N = K = 6$, $M_\mathrm{y} = 4$).}
    \label{fig:SR_Mtot_fixed_con}
\end{figure}

\begin{figure}
    \centering
    \subfigure[Idealized pattern]{
    \label{fig:SR_Mtot_idealized}
    \includegraphics[width = 0.48\textwidth]{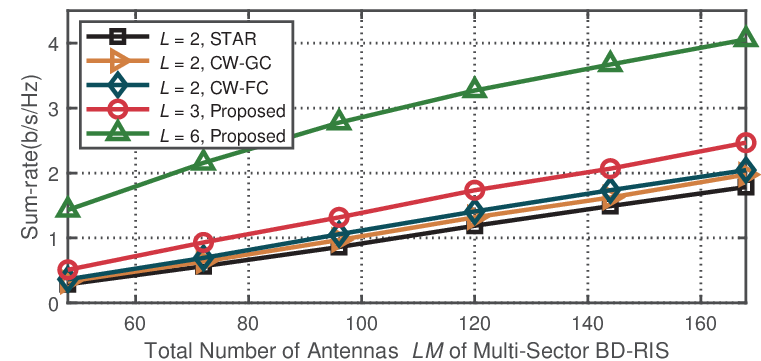}}
    \subfigure[Practical pattern]{
    \label{fig:SR_Mtot_practical}
    \includegraphics[width = 0.48\textwidth]{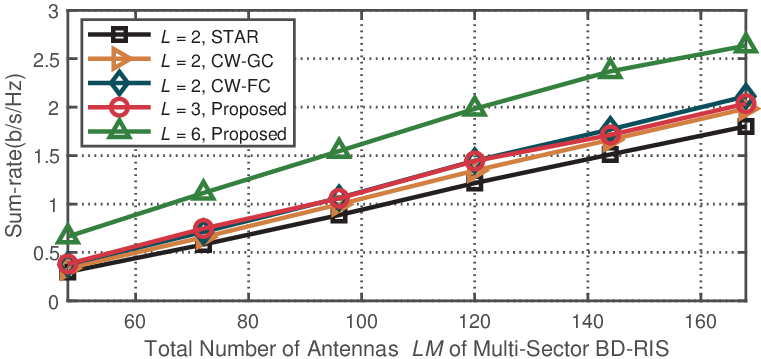}}
    \caption{Sum-rate versus total number of antennas of multi-sector BD-RIS ($N = K = 6$, $M_\mathrm{y} = 4$, $\kappa_o = 0$ dB, $\forall o\in\mathcal{O}$, $P_\mathrm{T} = 30$ dBm).}
    \label{fig:SR_Mtot_add_benchmark}
\end{figure}

\begin{figure*}
    \centering
    \subfigure[Idealized pattern ($\kappa_o = 0$ dB, $\forall o \in \mathcal{O}$, $L = 3$, $P_\mathrm{T} = 20$ dBm)]{
    \label{fig:SR_Mtot_dis_0_L3_iNPRP}
    \includegraphics[width = 0.48\textwidth]{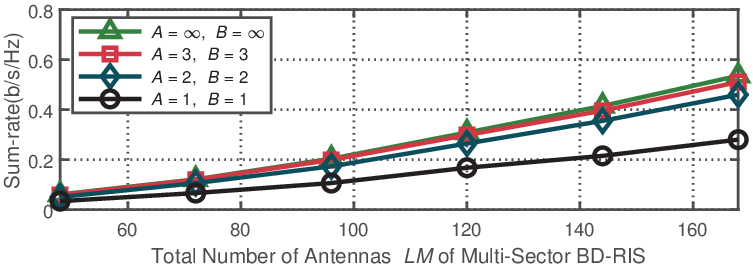}}
    \subfigure[Practical pattern ($\kappa_o = 0$ dB, $\forall o \in \mathcal{O}$, $L = 3$, $P_\mathrm{T} = 20$ dBm)]{
    \label{fig:SR_Mtot_dis_0_L3_pNPRP}
    \includegraphics[width = 0.48\textwidth]{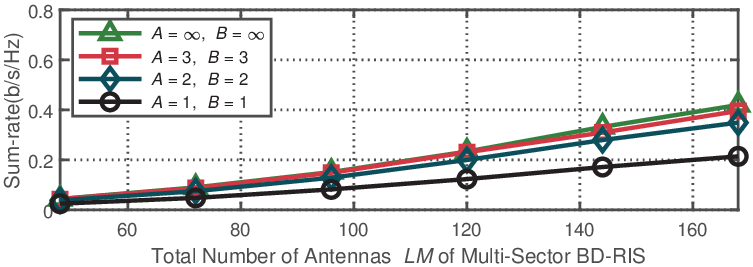}}
    \subfigure[Idealized pattern ($\kappa_o = 10$ dB, $\forall o \in \mathcal{O}$, $L = 3$, $P_\mathrm{T} = 10$ dBm)]{
    \label{fig:SR_Mtot_dis_10_L3_iNPRP}
    \includegraphics[width = 0.48\textwidth]{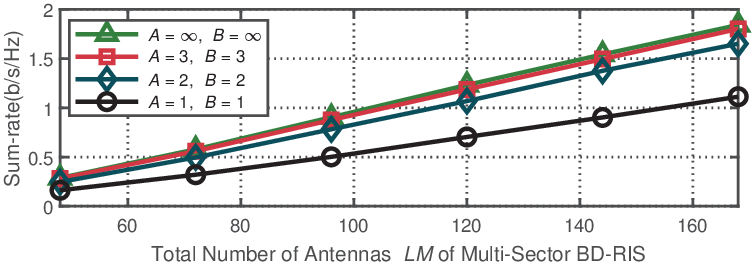}}
    \subfigure[Practical pattern ($\kappa_o = 10$ dB, $\forall o \in \mathcal{O}$, $L = 3$, $P_\mathrm{T} = 10$ dBm)]{
    \label{fig:SR_Mtot_dis_10_L3_pNPRP}
    \includegraphics[width = 0.48\textwidth]{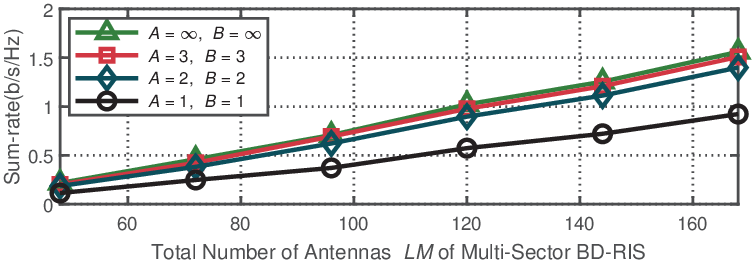}}   
    \subfigure[Idealized pattern ($\kappa_o = 0$ dB, $\forall o \in \mathcal{O}$, $L = 6$, $P_\mathrm{T} = 20$ dBm)]{
    \label{fig:SR_Mtot_dis_0_L6_iNPRP}
    \includegraphics[width = 0.48\textwidth]{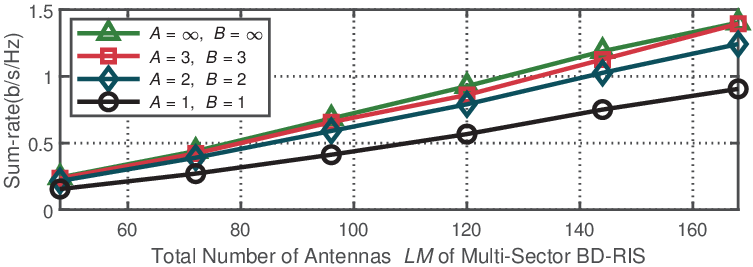}}
    \subfigure[Practical pattern ($\kappa_o = 0$ dB, $\forall o \in \mathcal{O}$, $L = 6$, $P_\mathrm{T} = 20$ dBm)]{
    \label{fig:SR_Mtot_dis_0_L6_pNPRP}
    \includegraphics[width = 0.48\textwidth]{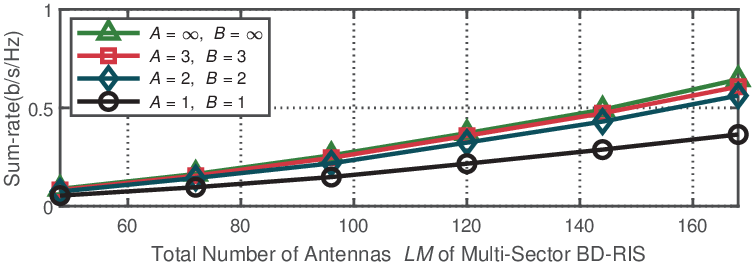}}
    \subfigure[Idealized pattern ($\kappa_o = 10$ dB, $\forall o \in \mathcal{O}$, $L = 6$, $P_\mathrm{T} = 10$ dBm)]{
    \label{fig:SR_Mtot_dis_10_L6_iNPRP}
    \includegraphics[width = 0.48\textwidth]{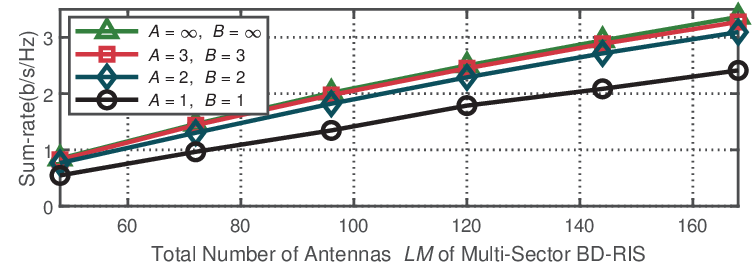}}
    \subfigure[Practical pattern ($\kappa_o = 10$ dB, $\forall o \in \mathcal{O}$, $L = 6$, $P_\mathrm{T} = 10$ dBm)]{
    \label{fig:SR_Mtot_dis_10_L6_pNPRP}
    \includegraphics[width = 0.48\textwidth]{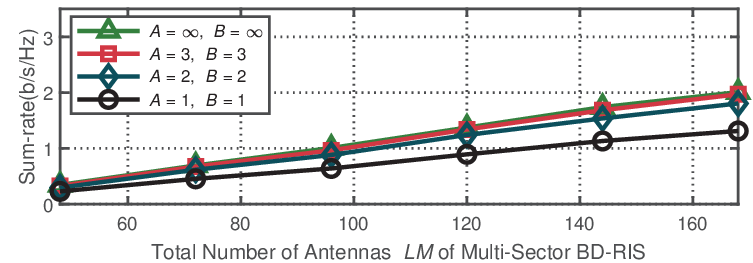}} 
    \caption{Sum-rate versus total number of antennas of multi-sector BD-RIS with different resolutions for amplitudes and phase shifts, different numbers of sectors, and different Rician fading factors ($N = K = 6$, $M_\mathrm{y} = 4$).}
    \label{fig:SR_Mtot_dis}\vspace{-0.2 cm}
\end{figure*}

Fig. \ref{fig:SR_Mtot_fixed_con} shows the sum-rate performance versus the total number of antennas of multi-sector BD-RIS. 
We consider different numbers of sectors and different Rician factors. 
We have the following observations from Fig. \ref{fig:SR_Mtot_fixed_con}.
1) Even with a reduced dimension for each sector, the sum-rate performance achieved by the multi-sector BD-RIS with an increasing number of sectors $L$ still enhances, which is in accordance with the conclusion in Fig. \ref{fig:P_L_Mtot_fixed}. This observation indicates that the higher antenna gain due to narrower beamwidth compensates for the reduced dimension of each sector, the reduced power split to each sector, and boosts the performance.
2) Given the same fading conditions, the multi-sector BD-RIS with a smaller number of antennas but a larger number of sectors can achieve similar performance to that with a larger number of antennas but a smaller number of sectors. 
For example, when the practical radiation pattern is adopted, the sum-rate performance with $ML = 120$ and $L = 3$ is similar to that with $ML = 96$ and $L = 6$ as shown in Fig. \ref{fig:SR_Mtot_fixed_con_0}. 
This phenomenon implies the number of antennas for the whole multi-sector BD-RIS can be reduced by increasing the number of sectors, which is meaningful for practical realizations of the multi-sector BD-RIS. 
3) With the fixed number of sectors, the sum-rate performance based on the idealized radiation pattern is always better than that based on the practical radiation pattern, which, again, matches with the result in Fig. \ref{fig:P_L_Mtot_fixed}.

In Fig. \ref{fig:SR_Mtot_add_benchmark}, we compare the sum-rate performance with existing RIS architectures. Specifically, STAR-RIS is essentially a special case of $L=2$ of our proposed model. We also add the performance of BD-RIS with $L=2$ and cell-wise group/fully-connected (CW-GC/FC) architectures proposed in \cite{li2022}, where the group size of CW-GC is fixed as 6, indicating that the CW-GC scheme and the proposed model with $L=6$ have the same circuit complexity. From Fig. \ref{fig:SR_Mtot_add_benchmark} we observe that when the idealized radiation pattern is adopted, the proposed multi-sector BD-RIS outperforms both STAR and CW-GC/FC schemes; when the practical radiation pattern is adopted, the proposed multi-sector BD-RIS outperforms STAR and CW-GC schemes, and achieves almost the same performance as the CW-FC scheme. More importantly, with the same circuit complexity of the $ML$-port reconfigurable impedance network, the proposed model (with $L=6$) can achieve much better sum-rate than the CW-GC scheme (with $L=2$). This observation highlights the benefit of antenna arrangements in boosting performance.

In Fig. \ref{fig:SR_Mtot_dis} we plot the sum-rate performance as a function of the total number of antennas of multi-sector BD-RIS when each nonzero element of the multi-sector BD-RIS has discrete amplitudes and phase shifts controlled by different resolutions. Similar conclusions as Fig. \ref{fig:SR_Mtot_fixed_con} can be obtained from Fig. \ref{fig:SR_Mtot_dis}. Moreover, we have the following additional observations. 
1) The sum-rate performance achieved by multi-sector BD-RIS with discrete values increases with the growth of resolutions. Specifically, when the idealized radiation pattern is applied for the modeling of fading channels, the sum-rate with $A = 3$ and $B = 3$ is close to that for the continuous case; when the practical radiation pattern is applied for the channel model, the sum-rate with $A = 2$ and $B = 2$ is sufficient to achieve satisfactory performance. 
2) When the total number of antennas of multi-sector BD-RIS is fixed, the performance gap between multi-sector BD-RIS with different resolutions become smaller with the growth of the number of sectors $L$.  
Recalling that the sum-rate performance increases with $L$, we can deploy the multi-sector BD-RIS with a larger number of sectors and lower resolutions to achieve the same performance as that with a smaller number of sectors and higher resolution. 
For example, the sum-rate performance with $L = 3$, $A = B = 2$, and $M = 40$ in Fig. \ref{fig:SR_Mtot_dis_0_L3_pNPRP} is almost the same as that with $L = 6$, $A = B = 1$, and $M = 20$ in Fig. \ref{fig:SR_Mtot_dis_0_L6_pNPRP}.
Such reduction of resolutions is beneficial for reducing the cost of multi-sector BD-RIS realizations.

\section{Conclusion and Future Work}
\label{sc:Conclusion}

In this paper, we propose a novel multi-sector BD-RIS model. Particularly, we derive the constraints of the cell-wise single-connected multi-sector BD-RIS for both continuous and discrete cases. We also derive the multi-sector BD-RIS-aided channel model including the impact of antenna gain. 
The existing STAR-RIS/IOS is a special instance of our proposed multi-sector BD-RIS model. 

With the proposed model, we derive the scaling law of the received power as a function of the number of sectors for the multi-sector BD-RIS-aided SU-SISO system. We show that there is a positive correlation between the maximum received power and the number of sectors. 

We then apply the multi-sector BD-RIS in the MU-MISO system and propose simple yet efficient algorithms to maximize the sum-rate of the multi-sector BD-RIS-aided MU-MISO system when amplitudes/phase shifts of the RIS have either continuous values or discrete values selected from finite codebooks. 
Simulation results show that given the same sum-rate requirement, the number of antennas for the whole multi-sector BD-RIS can be reduced by increasing the number of sectors, and that the codebook resolutions can also be reduced by increasing the number of sectors. 

Future research avenues include, but are not limited to the following aspects:

\subsubsection{The Dimensioning of Multi-Sector BD-RIS} The size and the number of antenna elements, and the number of sectors of multi-sector BD-RIS are jointly determined by the operating frequency of incident waves, the performance requirement, and circuit complexity/cost. Therefore, it is important to carefully determine the dimension of multi-sector BD-RIS to balance the performance and hardware complexity.

\subsubsection{Different Multi-Sector BD-RIS Design Approaches} While the sum-rate performance enhancement of the proposed multi-sector BD-RIS has been evaluated in this work, the benefit of multi-sector BD-RIS could also be investigated from various perspectives by considering different algorithms and different metrics, such as circuit complexity minimization and power minimization.

\begin{appendices}

\section{Proof of Lemma 1}
By applying the Lagrange method, we can easily obtain $(\mathbf{\Psi}_m + \mu\mathbf{I}_L){\bm \beta}_m = {\bm \chi}_m$ with ${\bm \beta}_m^T{\bm \beta}_m = 1$, where $\mu \in \mathbb{R}$ is the Lagrange multiplier. 
Now it remains to show $\mathbf{\Psi}_m + \mu\mathbf{I}_L$ is positive semi-definite. 
Suppose ${\bm \beta}_m'$ is a feasible solution of problem (\ref{eq:amplitude_cell_m_continuous1}) with 
${\bm \beta}_m'^T{\bm \beta}_m' = 1$. Then we have  $f({\bm \beta}_m') \ge f({\bm\beta}_m^{\star})$. Given that $(\mathbf{\Psi}_m + \mu\mathbf{I}_L){\bm \beta}_m^\star = {\bm \chi}_m$, we have
\begin{equation}
    \label{eq:inequ1}
    \begin{aligned}
        &{\bm \beta}_m'^T\mathbf{\Psi}_m{\bm \beta}_m' - 2{\bm \beta}_m^{\star T}(\mathbf{\Psi}_m + \mu\mathbf{I}_L){\bm \beta}_m'\\
        &~~~ \ge  {\bm \beta}_m^{\star T}\mathbf{\Psi}_m{\bm \beta}_m^\star - 2{\bm \beta}_m^{\star T}(\mathbf{\Psi}_m + \mu\mathbf{I}_L){\bm \beta}_m^\star.
    \end{aligned}
\end{equation}
Rearranging (\ref{eq:inequ1}), we have
\begin{equation}
    \label{eq:inequ2}
    \begin{aligned}
        &({\bm \beta}_m'^T - {\bm \beta}_m^{\star T})(\mathbf{\Psi}_m + \mu\mathbf{I}_L)({\bm \beta}_m'-{\bm\beta}_m^\star) \\
        &~~~~~~ \ge \mu({\bm \beta}_m'^T{\bm \beta}_m' - {\bm \beta}_m^{\star T}{\bm\beta}_m^\star) = 0,
    \end{aligned}
\end{equation}
which completes the proof.

\section{Proof of Lemma 2}
With $\varepsilon_{m,l} + \mu \ge \varepsilon_{m,1} + \mu> 0$, we have 
\begin{equation}
    \begin{aligned}
        1 = &g(\mu) = \sum_{l\in \mathcal{L}}\frac{\tilde{\chi}_{m,l}^2}{(\varepsilon_{m,l} + \mu)^2}\\
         &\le \sum_{l\in \mathcal{L}}\frac{\tilde{\chi}_{m,l}^2}{(\varepsilon_{m,1} + \mu)^2} = \frac{\|{\bm \chi}_m\|_2^2}{(\varepsilon_{m,1} + \mu)^2}, 
    \end{aligned}
\end{equation}
which yields $\mu\le \|{\bm \chi}_m\|_2 - \varepsilon_{m,1}$. Meanwhile, we have 
\begin{equation}
    \begin{aligned}
        1 = &g(\mu) = \sum_{l\in \mathcal{L}}\frac{\tilde{\chi}_{m,l}^2}{(\varepsilon_{m,l} + \mu)^2}\\
         &\ge \sum_{l\in \mathcal{E}}\frac{\tilde{\chi}_{m,l}^2}{(\varepsilon_{m,l} + \mu)^2} = \frac{\sum_{\forall l \in \mathcal{E}}\tilde{\chi}_{m,l}^2}{(\varepsilon_{m,1} + \mu)^2}, 
    \end{aligned}
\end{equation}  
and thus $\mu \ge \sqrt{\sum_{\forall l \in \mathcal{E}}\tilde{\chi}_{m,l}^2} - \varepsilon_{m,1}$.

\end{appendices}

\bibliographystyle{IEEEtran}
\bibliography{references}

\end{document}